\documentclass[a4paper]{article}
\usepackage[a4paper,top=3cm,left=3cm,right=2cm,bottom=2cm]{geometry}
\usepackage{amsmath}
\usepackage{amsthm}
\usepackage{natbib}
\usepackage{algorithm}
\usepackage{setspace}
\usepackage{fancyhdr}
\usepackage{amssymb}
\usepackage{setspace}
\usepackage[usenames,dvipsnames]{xcolor}

\usepackage{lscape}
\usepackage{subfig}
\usepackage{graphics}
\usepackage{graphicx}
\usepackage{caption}
\usepackage{csquotes}

\usepackage[noend]{algpseudocode}

\usepackage{longtable}
\usepackage{multirow}
\usepackage{url}
\usepackage[vertfit]{breakurl} 
\usepackage{booktabs}
\newtheorem{proposition}{Proposition}
\usepackage{array}
\usepackage{tikz}
\usepackage{rotating}
\usepackage{morefloats}
\usepackage{colortbl}
\usepackage[normalem]{ulem}
\usepackage[titletoc,title]{appendix}

\usetikzlibrary{arrows}
\newcolumntype{L}[1]{>{\raggedright\let\newline\\\arraybackslash\hspace{0pt}}m{#1}}
\newcolumntype{C}[1]{>{\centering\let\newline\\\arraybackslash\hspace{0pt}}m{#1}}
\newcolumntype{R}[1]{>{\raggedleft\let\newline\\\arraybackslash\hspace{0pt}}m{#1}}

\newcommand{\rev}[1]{{#1}}



\begin{document}
\sloppy
\allowdisplaybreaks
\newcolumntype{H}{>{\setbox0=\hbox\bgroup}c<{\egroup}@{}}




\large
\title{Enhanced arc-flow formulations to minimize weighted completion time on identical parallel machines}

\author{
{\bf Arthur Kramer, Mauro Dell'Amico, Manuel Iori}\\
Dipartimento di Scienze e Metodi dell'Ingegneria\\ Universit\`{a} degli Studi di Modena e Reggio Emilia, Italy\\
arthur.kramer@unimore.it, mauro.dellamico@unimore.it, manuel.iori@unimore.it\\ \\
}

\date{}

\maketitle

\vspace{-0.5cm}
\begin{center}
\end{center}

\begin{abstract}
	{We consider the problem of scheduling a set of jobs on a set of identical parallel machines, with the aim of minimizing the total weighted completion time. The problem has been solved in the literature with a number of mathematical formulations, some of which require the implementation of tailored branch-and-price methods. In our work, we solve the problem instead by means of new arc-flow formulations, by first representing it on a capacitated network and then invoking a mixed integer linear model with a pseudo-polynomial number of variables and constraints.
	According to our computational tests,  existing formulations from the literature can solve to proven optimality benchmark instances with up to 100 jobs, whereas our most performing arc-flow formulation solves all instances with up to 400 jobs and provides very low gap for larger instances with up to 1000 jobs.}
\end{abstract}

%
\onehalfspace
\section{Introduction}

We are given a set $J = \{1, 2, \dots,n\}$ of jobs to be scheduled on a set $M = \{1, 2, \dots,m\}$ of identical parallel machines. Each job $j \in J$ has a processing time $p_j$ and a penalty weight $w_j$. A schedule is feasible if each job is assigned to a unique machine and processed without preemption, and each machine processes at most one job at a time. Let $C_j$ define the completion time of job $j$, our goal is to find a feasible schedule for which the total weighted completion time, $\sum_{j=1}^{n} w_j C_j$, is a minimum. The problem is denoted as $P||\sum w_jC_j$ as in the three-field classification of \cite{Graham1979}. In the following, we suppose that processing times and penalty weights take integer values.

The $P||\sum w_jC_j$ was proven to be $\mathcal{NP}$-hard by \cite{Bruno1974}. Despite being a classical production scheduling problem, with real-world applications, it has not received much attention in {the literature} and cannot be considered a well solved problem.
To the best of our knowledge, state-of-the-art exact methods for the $P||\sum w_jC_j$ are the branch-and-bound algorithms developed by \cite{AzizogluandKirca19999b}, \cite{ChenandPowell1999} and \cite{vandenAkkeretal1999}, the last two of which make use of column generation techniques {to solve the relaxed problem at each node}. Aside from these works, the $P||\sum w_jC_j$ can be solved by adapting mathematical formulations originally developed for similar one machine or unrelated parallel machines problems. Among these formulations, we mention the {\em time indexed} (TI) \emph{mixed integer linear programming} (MILP) model  by \cite{SousaWolsey1992}, the {convex integer quadratic programming} model  by \cite{Skutella2001} {and the preemptive TI model by \cite{BulbulandSen2017}. According to our tests, these methods fail in solving some $P||\sum w_jC_j$ benchmark instances involving just 100 jobs.
	
In this paper, we solve exactly large-size instances of the $P||\sum w_jC_j$ by focusing on the development of  \emph{arc-flow} (AF) formulations. AF formulations represent the problem as a capacitated network with side constraints, and consist of a MILP model with a pseudo-polynomial number of variables and constraints.
AF formulations have been used to model many combinatorial optimization problems (see, e.g., \citealt{W77} and \citealt{ValeriodeCarvalho1999}), and have recently obtained successful results on important areas such as bin packing and cutting stock problems (see, e.g., the recent survey by \citealt{Delorme2016}).
\rev{For the area of scheduling, we are only aware of a very recent publication by \cite{MradAndSouayah2018} that presents an AF formulation for the problem of minimizing makespan on identical parallel machines.} In our work, we first propose a straight AF formulation, and then enhance it through a set of techniques that aim at reducing the number of variables and constraints by combining established \rev{reduction procedures} from the literature with some specific features of the $P||\sum w_j C_j$. This results in a powerful method that solves to proven optimality {large instances and provide low optimality gaps for very large instances.}

The remainder of this paper is organized as follows. In the next Section \ref{sec:litreview}, we review the main literature, whereas in Section \ref{sec:mathform} we adapt to the $P||\sum w_j C_j$ some mathematical formulations {from the literature}. In Section \ref{sec:AF}, we present the straight and enhanced AF formulations. In Section \ref{sec:results}, we provide the outcome of extensive computational experiments and finally, in Section \ref{sec:conclusion}, we present some concluding remarks.

\section{Literature review} \label{sec:litreview}

{The $P||\sum w_jC_j$ is a generalization of  $1||\sum C_j$,  $1||\sum w_jC_j$ and  $P||\sum C_j$, which are all solvable in polynomial time using the well known \emph{shortest processing time} (SPT) rule, or the {\em weighted shortest processing time} (WSPT) rule of \cite{Smith1956}. The WSPT rule sorts jobs according to non-increasing $w_j/p_j$. The $P||\sum w_jC_j$, on the contrary, is a difficult problem, and was proven to be $\mathcal{NP}$-hard even with just two machines ($P2||\sum w_jC_j$, see \citealt{Bruno1974}).}

The literature on the $P||\sum w_jC_j$ focused on the development of early heuristic methods and exact {\em branch-and-bound} (B\&B) algorithms. \cite{Eastman1964} proposed a heuristic that uses a variant of the WSPT rule. \cite{Kawaguchi1986} showed that such heuristic guarantees a solution whose total weighted completion time is not worse than $(\sqrt{2}+1)/2$ times the optimal solution value. \cite{ElmaghrabyPark1974} proposed a B\&B based on the use of lower bounds and properties of optimal solutions.
\rev{\cite{Sarinetal1988} improved the work by \cite{ElmaghrabyPark1974}, by proposing a new branching scheme that substantially reduces the number of schedules to be evaluated in the  B\&B tree.} They solved instances with up to 30 jobs and 5 machines. \cite{BelouadahandPotts1994} incorporated Lagrangian relaxation in a B\&B based on a TI formulation, solving instances with up to 30 jobs and 8 machines.
Another B\&B was designed by \cite{AzizogluandKirca19999b}, who used the same branching scheme of \cite{Sarinetal1988}, but enriched it with the lower bound of \cite{Webster1995}, solving instances with up to $35$ jobs and $5$ machines.

\cite{ChenandPowell1999}  tackled the $P||\sum w_j C_j$, the $Q||\sum w_j C_j$ and the $R||\sum w_j C_j$ by means of a \emph{set covering} (SC) formulation where each column corresponds to a single machine schedule. To deal with the large number of feasible schedules, they developed a {\em branch-and-price} (B\&P) method in which at each node of an enumeration tree a valid lower bound was obtained by column generation. In the same year, \cite{vandenAkkeretal1999} independently developed a similar B\&P. They focused only on the $P||\sum w_j C_j$, and obtained slightly better results than \cite{ChenandPowell1999} by branching on completion times instead of branching on variables that indicate whether a certain job is processed immediately after another job. Their B\&P solved $P||\sum w_j C_j$ instances with up to 100 jobs and 10 machines. \rev{Very recently, \cite{KowalczykandLeus2018} extended the method of \cite{vandenAkkeretal1999} by investigating the use of stabilization techniques, a generic branching rule (see \citealt{RyanandFoster1981}), and a {zero-suppressed binary decision diagram} approach (see \citealt{Minato1993}) for solving the pricing subproblem.}

\cite{SousaWolsey1992} proposed a TI formulation to solve single machine scheduling problems with general objective function. Their formulation originates from early works on scheduling (see, e.g., \citealt{Bowman1959} and \citealt{PritskerEtAl1969}) and is easily adaptable to multiple machine problems.

Recent literature focused on the related $R||\sum w_jC_j$, which considers unrelated parallel machines. \cite{Skutella2001} proposed a {\emph{convex integer quadratic programming} (CIQP)} relaxation as the basis for an approximation algorithm. In his approach, the problem is  formulated as an integer quadratic programming model with $n \times m$ assignment variables. Then, integrality is relaxed and the objective function is convexified to obtain a CIQP relaxation that can be solved in polynomial time. Finally, the relaxed solution is transformed in a feasible solution by applying a randomized rounding method. Later on, \cite{PlateauandRiosSolis2010} embedded the CIQP relaxation  by \cite{Skutella2001} in a branch-and-bound algorithm to obtain an exact approach.
\cite{BulbulandSen2017} proposed a Benders decomposition method based on a TI formulation. Their formulation accepts preemptive solutions, but is proven to yield non-preemptive optimal solutions because of the use of tailored coefficients in the objective function. Their method obtained better results than those produced by the CIQP formulation.
We also highlight the recent works by \cite{Rodriguezetal2012, Rodriguezetal2013}, who developed metaheuristic algorithms based on {GRASP} and iterated greedy paradigms, \rev{and tested them on instances having either unrelated or uniform machines. Note that we cannot compare with the literature on these instances, as our methods are specifically tailored for the case of identical machines.}

We conclude this section by referring the interested reader to the review by \cite{LiandYang2009} on models, relaxations and algorithms for minimizing weighted completion times on parallel machines. Another review of models for parallel machines scheduling problems, which includes a computational evaluation of MILP models, was proposed by \cite{UnluandMason2010}. They classified the formulations into four different types according to the characteristics of their variables (TI variables, network variables, assignment variables, and positional date variables), and concluded that TI formulations tend to perform better than the others. \rev{A very recent review on preemptive models for scheduling problems with controllable processing times has been presented by \cite{SSS18}, who also included a section on methods based on flow computations.}

\section{Existing mathematical formulations}
\label{sec:mathform}

In this section, we provide an overview of mathematical formulations for the $P||\sum w_jC_j$ that we obtained by adapting models originally presented for related problems.

\subsection{Sousa and Wolsey's time indexed formulation}\label{TI}

The TI formulation by \cite{SousaWolsey1992} was originally designed to deal with the problem of sequencing jobs over the time on a single machine subject to resource constraints. As previously highlighted, this formulation can be modified to consider parallel machines and weighted completion time as follows:
\begin{align}
(\mbox{TI}) \quad \min {\sum_{j \in J} \sum_{t = 0}^{T-p_j} w_j t x_{jt}} + \sum_{j \in J} w_j p_j \label{FO:TI}\\
\text{st.} \sum_{t = 0}^{T - p_j} x_{jt} = 1 & &&  j \in J  \label{constr1:TI}\\
\sum_{j \in J} { \sum_{s = \max\{0, t+1-p_j\}}^{\min\{t,T+1-p_j\}} } x_{i s} \leq m & &&  t = 0,\dots,T-1 \label{constr2:TI}\\
x_{jt} \in \{0,1\} & &&  j \in J,   t = 0, \dots, T-p_j \label{constr3:TI}
\end{align}
where $x_{jt}$ is a binary decision variable taking value $1$ if job $j$ starts its processing at time $t$, $0$ otherwise. The time horizon is defined by $T$ and should be sufficiently large to ensure optimality {and as short as possible to avoid the creation of unnecessary variables}. The objective function \eqref{FO:TI} seeks the minimization of the total weighted completion time. Note that we expressed the completion time of a job as the sum of starting time and processing time, formally using $C_j = \sum_t t x_{jt} + p_j$. The term $\sum_j w_jp_j$ is a constant and is thus irrelevant for the formulation. Constraints \eqref{constr1:TI} ensure that each job is processed exactly once. Constraints \eqref{constr2:TI} forbid overlapping among the jobs {by imposing that at most $m$ jobs are executed in parallel at any time. Constraints} \eqref{constr3:TI} define the variables' domain. Model \eqref{FO:TI}--\eqref{constr3:TI} contains a pseudo-polynomial number of variables, a common characteristic of TI formulations, which amounts to $\mathcal{O}(nT)$.

\subsection{Skutella's convex integer quadratic programming formulation} \label{sec:CQP}

The idea behind the method of \cite{Skutella2001} is to formulate the $R||\sum w_jC_j$ as an integer quadratic program and then convexify the objective function.
His formulation uses $n \times m$ integer assignment variables and can be adapted to  the $P||\sum w_jC_j$ as follows:
\begin{align}
(\mbox{CIQP}) \quad \min \sum_{j\in J} w_j C_j \label{FO:CQP}\\
\text{st.} \sum_{k \in M} x_{j}^{k} = 1 & &&  j \in J \label{constr1:CQP}\\
C_j = \sum_{k \in M} x_{j}^{k} \left( \frac{1+x_{j}^{k}}{2}p_j + \sum_{i \in J, i \prec j}x_{i}^{k}p_i \right) & &&  j \in J \label{constr2:CQP}\\
x_{j}^{k} \in \{0,1\} & &&  j \in J,  k \in M\label{constr3:CQP}
\end{align}
where $x_j^k = 1$ if job $j$ is scheduled on machine $k$, $0$ otherwise. The notation $i \prec j$ in \eqref{constr2:CQP} means that either $(w_i/p_i > w_j/p_j)$ or $(w_i/p_i = w_j/p_j$ and $i < j)$. This is used to take into account that the jobs are scheduled on each machine by non-increasing order of $w_j/p_j$, i.e., by following the WSPT rule. The relaxation obtained by dropping integrality constraints from \eqref{constr3:CQP} can be solved easily. Indeed, \cite{Skutella2001} showed that $x_j^k = 1/m$ $\forall j \in J$ and $k \in M$ is an optimal solution to this relaxation when the machines are identical. The optimal integer solution {of model \eqref{FO:CQP}--\eqref{constr3:CQP}} can be obtained by invoking a commercial CIQP solver {such as CPLEX or Gurobi (as done by \citealt{PlateauandRiosSolis2010} for the $R||\sum w_jC_j$).
	
\subsection{B{\"u}lb{\"u}l and {\c{S}}en's preemptive time indexed formulation}\label{sec:PeemptionTI}

\cite{BulbulandSen2017} modeled the preemptive version of the $R||\sum w_jC_j$ by means of a \emph{preemptive time indexed} (PTI) formulation making use of $n \times m \times T$ continuous variables and $n \times m$ binary variables. Then, they proved that it is always possible to devise a non-preemptive solution having the same objective value of the optimal PTI preemptive one, concluding that the PTI model is optimal also for the (non-preemptive) $R||\sum w_jC_j$. Their formulation can  be adapted to the $P||\sum w_jC_j$ as follows:
\begin{align}
(\mbox{PTI}) \quad \min \sum_{j \in J} \sum_{t = 1}^{T} \sum_{k \in M} \frac{w_j}{p_j}\left(t + \frac{p_j - 1}{2}\right) x_{jkt} \label{FO:PTI}\\
\text{s.t. } \sum_{t = 1}^{T} x_{jkt} = p_j y_{jk} & &&  j \in J,  k \in M \label{constr1:PTI}\\
\sum_{j \in J} x_{jkt} \leq 1 & &&  k \in M, t = 1,\dots,T \label{constr2:PTI}\\
\sum_{k \in M} y_{jk} = 1 & &&  j \in J \label{constr3:PTI}\\
x_{jkt} \geq 0 & &&  j \in J,  k \in M, t = 1,\dots,T \label{constr4:PTI}\\
y_{jk} \in \{0,1\} & &&  j \in J,  k \in M \label{constr5:PTI}
\end{align}
where $x_{jkt}$ are continuous variables representing the quantity of job $j$, i.e., the number of unit-length parts of job $j$, which is finished at time $t$ on machine $k$, and $y_{jk}$ takes value $1$ if job $j$ is assigned to machine $k$, $0$ otherwise. Constraints \eqref{constr1:PTI} state that all unit-length parts of job $j$ must be processed on the same machine. Constraints \eqref{constr2:PTI} ensure that each machine processes at most one job at a time. Constraints \eqref{constr3:PTI} guarantee that each job is assigned to exactly one machine. Constraints \eqref{constr4:PTI} and \eqref{constr5:PTI} give the variables' domains.

To tackle the pseudo-polynomial number of variables and constraints in the aforementioned model, the authors proposed a Benders decomposition approach. Their idea is to start by first solving a master problem composed by variables $y$ and constraints \eqref{constr3:PTI} and \eqref{constr5:PTI}, and obtain a solution $\bar{y}$. Then, solving a set of $m$ subproblems, each for a machine, that use the $\bar{y}$ solution but only involve variables $x$ and constraints \eqref{constr1:PTI}, \eqref{constr2:PTI} and \eqref{constr4:PTI}. The subproblems either deliver an optimal $P||\sum w_jC_j$ solution, or some optimality Benders cuts to be added to the master problem. This process is reiterated until proof of optimality or some stopping criteria are met.

\subsection{Set covering formulation} \label{sec:SC}

\emph{Set covering} (SC) formulations are widely used to model combinatorial optimization problems as covering problems. \cite{vandenAkkeretal1999} followed this idea and modeled the $P||\sum w_jC_j$ by using an SC formulation having an exponential number of variables. Let $S$ be a set containing all feasible schedules for a single machine, let $a_{js}$ be a binary coefficient indicating whether job $j \in J$ is included or not in schedule $s$, and let $x_s$ be a binary variable assuming value $1$ if schedule $s \in S$ is selected, $0$ otherwise. The $P||\sum w_jC_j$ can be modeled as:
\begin{align}
(\mbox{SC}) \quad \min \sum_{s \in S} c_s x_s \label{FO:SC}\\
\text{s.t. } \sum_{s \in S} x_s = m  \label{constr1:SC}\\
\sum_{s \in S} a_{js} x_{s} = 1 & &&  j \in J \label{constr2:SC}\\
x_s \in \{0,1\} & &&  s \in S \label{constr3:SC}
\end{align}
Constraints \eqref{constr1:SC} state that exactly $m$ schedules are selected. Constraints \eqref{constr2:SC} ensure that each job is processed once, and constraints \eqref{constr3:SC} impose variables to be binary. As model \eqref{FO:SC}--\eqref{constr3:SC} has an exponential number of variables, the authors solved it with a branch-and-price algorithm. In particular, they solved each node of an enumeration tree by means of a column generation method, which looks for negative cost schedules by invoking a tailored \rev{ {\it dynamic programming} (DP) algorithm}. They performed branching by considering the minimum completion time of a fractional job. Let $\bar x$ be the current solution and $\bar S \subseteq S$ the set of schedules associated with positive $\bar x$ values. A fractional job is defined as a job $j$ for which $\sum_{s \in \bar S} C_j(s) \bar x_s > \min \{C_j(s) | \bar x_s > 0\}$, where $C_j(s)$ is the completion time of job $j$ in schedule $s$. Note that \cite{ChenandPowell1999} also proposed a branch-and-price algorithm to solve model \eqref{FO:SC}--\eqref{constr3:SC}, but, differently from \cite{vandenAkkeretal1999}, they performed branching directly on the $x$ variables.

\section{Arc-flow formulations} \label{sec:AF}
We first present a straight formulation, and then enhance \rev{it} with reduction procedures.

\subsection{Straight arc-flow formulation}
\label{subsec:SAF}

AF formulations are an established combinatorial optimization technique that models problems by using flows on a capacitated network (see, e.g., \citealt{W77}). When applied to machine scheduling problems, \rev{the flow obtained by solving an AF formulation} can be easily decomposed into paths (see \citealt{AMO93}), so that each path corresponds to a schedule of  activities on a machine. AF formulations make use of a pseudo-polynomial number of variables and constraints, and are thus related to the TI formulations that we previously described. The research effort behind AF is, however, to try to reduce as much as possible the required number of variables and constraints, thus keeping the size of the model as small as possible while preserving optimality.
In our work, we follow the recent literature on AF, that, starting from \cite{ValeriodeCarvalho1999}, used these techniques to obtain good computational results on cutting and packing problems  (see, e.g., \citealt{Delorme2016} for an updated survey), but we take into account issues that are typical of the scheduling field.

Our AF formulation models the $P||\sum w_jC_j$ as the problem of finding $m$ independent paths that start from a source node $0$, end at a destination node $T$, and  cover all the jobs. {For the sake of clarity, we start by presenting the very basic model, and focus later (in Algorithm 	\ref{alg:PandA}) on a first reduction of variables and constraints that is based on the WSPT sorting.
Our very basic AF formulation} uses a direct acyclic multigraph $G = (N,A)$. The set of vertices $N \subseteq \{0, 1, \dots, T\}$ can be initially considered as the set of {\em normal patterns} (for which we refer to the seminal papers by  \citealt{H72} and \citealt{CW77}, and to the recent discussion in \citealt{CI16}), i.e., the set of all the feasible combinations of jobs' processing times whose resulting value is between 0 and $T$. Let  $J_{+} = J \cup \{0\}$ include the original set of jobs plus a dummy job $0$ {having $p_0 = 0$ and $w_0 = 0$.} The set of arcs is partitioned as $A = \cup_{j \in J_{+}} A_{j}$. Each $A_{j}$ represents the set of {\em jobs arcs} associated with job $j \in J$, and, {for the moment, let us define it} as $A_j = \{ (q,r,j): {r-q=p_j} \text{ and } q \in N \}$. In addition, $A_0$ represents the set of {\em loss arcs}, that are used to model the amount of idle time between the end of activities and $T$ on a machine, and is defined as $A_0 = \{ (q,T,0): q \in N\}$. Let us also use $\delta^{+}(q) \subseteq A$, respectively $\delta^{-}(q) \subseteq A$, to define the subset of arcs that emanate from, respectively enter, a given node $q \in N$. A feasible $P||\sum w_jC_j$ solution can be represented as a set of $m$ paths in $G$, each corresponding to a machine schedule that start in 0 and make use of jobs arcs and of possibly one last loss arc to reach $T$.

To formulate the $P||\sum w_jC_j$  as an AF, we associate with each job arc $(q,r,j) \in A$ a variable $x_{qrj}$ that has a twofold meaning: for jobs arcs $(q,r,j) \in A_j$, $x_{qrj}$ takes value 1 if job $j$ is scheduled at start time $q$, 0 otherwise; for loss arcs $(q,T,0) \in A_0$, $x_{qT0}$ gives the \rev{number} of paths that end with arc $(q,T,0)$, i.e., that contain activities that finish at time $q$. The $P||\sum w_jC_j$ can then be modeled as:
\begin{align}
(\mbox{AF}) \quad \min \sum_{(q,r,j) \in A} w_j q x_{qrj} + \sum_{j \in J} w_j p_j \label{FO:AF}\\
\sum_{(q,r,j) \in \delta^{+}(q)} x_{qrj} - \sum_{(p,q,j) \in \delta^{-}(q)} x_{pqj} = \left\{
\begin{array}{l l}
m, & \quad \text{ if $q = 0$}\\
-m, & \quad \text{ if $q = T$}\\
0, & \quad \text{otherwise}
\end{array}\right. & &&  q \in N\label{constr1:AF}\\
\sum_{(q,r,j) \in A} x_{qrj} \geq 1 & &&  j \in J \label{constr2:AF}\\
x_{qrj} \in \{0,1\} & &&  (q,r,j) \in A \setminus A_0 \label{constr3:AF}\\
0 \leq x_{qT0} \leq m & && { (q,T,0) \in A_0} \label{constr4:AF}
\end{align}
The objective function \eqref{FO:AF} minimizes the sum of the weighted completion times. Constraints \eqref{constr1:AF} impose both flow conservation at each node and the use of exactly $m$ paths. Constraints \eqref{constr2:AF} impose all jobs to be scheduled, whereas constraints \eqref{constr3:AF} and \eqref{constr4:AF} give the variables' domains. Note that variables associated to loss arcs do not need to be defined as integers.

A first, simple but very important rule can be used to decrease the number of {variables and constraints in the model}. As previously discussed, in any optimal solution of the $P||\sum w_jC_j$ the jobs are sequenced on each machine by following the WSPT rule. Consequently, only arcs fulfilling this sorting rule can be considered (the first job in the order can only start in 0, the second job can start in 0 or right after the first job, and so on). The procedure that we implemented to build the underlying AF multi-graph, given in Algorithm \ref{alg:PandA}, takes this fact into consideration, producing a first reduction {of the size of sets $N$ and $A$, and hence of the number of variables and constraints in the formulation}. The procedure initializes the set of nodes and arcs to the empty set. It then considers one job $j$ at a time, according to the WSPT rule, to create the $A_j$ sets at steps 4--7. The sets of nodes and of loss arcs are constructed at steps 8--9, and the overall set of arcs at step 10.

\begin{algorithm}[htb]
	\caption{Construction of the AF multi-graph}\label{alg:PandA}
	\begin{algorithmic}[1]
		\Procedure {Create\_Patterns\_and\_Arcs}{$T$}
		\State {\bf{initialize}} $P[0 \dots T] \gets 0;$ {\footnotesize {\color{gray} \Comment $P$: array of size $T+1$}}
		\State {\bf{initialize}} $N \gets \emptyset; A[0 \dots n] \gets \emptyset;$ {\footnotesize {\color{gray} \Comment $N$: set of {vertices}; $A$: set of arcs}}
		\State $P[0] \gets 1;$
		\For {$j \in J$ according to the WSPT rule}
		\For {$t \gets T - p_j$ down to $0$}
		\If {$P[t] = 1$}
		$P[t + p_j] \gets 1; \rev{A[j] \gets A[j]}  \cup \{(t, t + p_j,j)\}; $ {\footnotesize {\color{gray} \Comment \rev{$A[j]$}: set of job arcs of $j$}}
		\EndIf
		\EndFor
		\EndFor
		\For {$t \gets 0$ to  $T$}
		\If {$P[t] = 1$} {$N \gets N \cup \{t\}$}; $\rev{A[0] \gets A[0]} \cup \{(t, T, 0)\}$ {\footnotesize {\color{gray} \Comment \rev{$A[0]$}: set of loss arcs}} \EndIf
		\EndFor
		\State $A \gets \cup_{j \in J+} \rev{A[j]}$
		\State \bf{return} $N, A$
		\EndProcedure
	\end{algorithmic}
\end{algorithm}

To ease comprehension, we present a simple example with 4 jobs, 2 machines, and $T = 8$ {(\rev{details on how to compute a strict value of $T$ are given} in Section \ref{sec:method} below)}. The characteristics of the 4 jobs and their sorting according to the WSPT rule are given in Figure \ref{fig:Ex1Alg1a}. The AF multi-graph built by Algorithm \ref{alg:PandA} is given in Figure \ref{fig:Ex1Alg1b}, and contains 9 vertices, 11 job arcs, and 7 loss arcs. An optimal solution is provided in Figure \ref{fig:Ex1Alg1c}, where we highlight the 2 paths corresponding to the machine schedules.

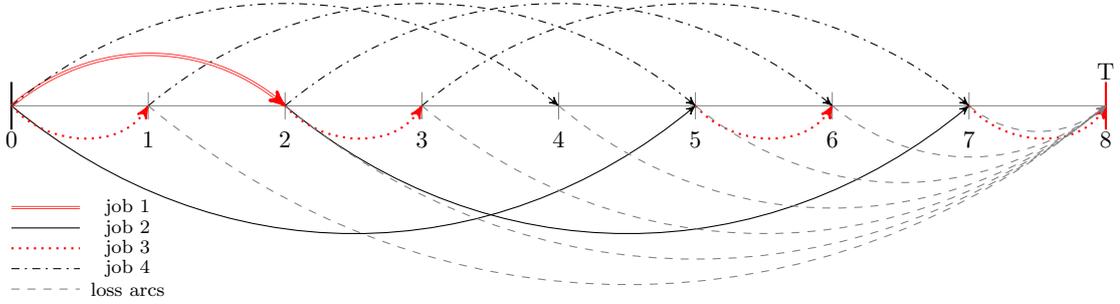
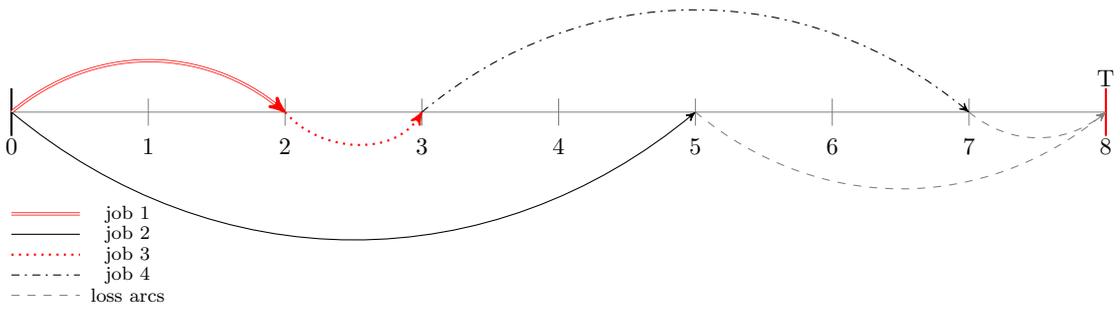
\begin{figure}[htb]
	\centering
	\subfloat[Input data ($T=8$)]{
		\label{fig:Ex1Alg1a}
		\footnotesize
		\begin{tabular}{ccccH}
			\toprule
			$j$ & $p_{j}$ & $w_{j}$ & $w_{j}/p_{j}$ &{WSPT rank} \\
			\cmidrule(r){1-1}
			\cmidrule(r){2-5}
			1 & 2 & 4 & 2.00 & 1\\
			2 & 5 & 7 & 1.40 & 2\\
			3 & 1 & 1 & 1.00 & 3\\
			4 & 4 & 3 & 0.75 & 4\\
			\bottomrule
		\end{tabular}
	}\\
	\subfloat[AF formulation]{
		\label{fig:Ex1Alg1b}
		\centering
		\begin{tikzpicture}[thick, scale=0.90, every node/.style={transform shape}, >=stealth', dot/.style = {draw, fill = white, circle, inner sep = 0pt, minimum size = 4pt}]
		\draw[very thin, color=black!50] (0,0) -- (16,0);
		\node (nT) at (16,0.5){T};\draw[color=red] (16,0.35) -- (16,-0.35);
		\node (n0) at (0,-0.5){$0$}; \draw[very thin, color=black!50] (0,0.2) -- (0,-0.2);\draw[color=black] (0,0.35) -- (0,-0.35);
		\node (n1) at (2,-0.5){$1$};\draw[very thin, color=black!50] (2,0.2) -- (2,-0.2);
		\node (n2) at (4,-0.5){$2$};\draw[very thin, color=black!50] (4,0.2) -- (4,-0.2);
		\node (n3) at (6,-0.5){$3$};\draw[very thin, color=black!50] (6,0.2) -- (6,-0.2);
		\node (n4) at (8,-0.5){$4$};\draw[very thin, color=black!50] (8,0.2) -- (8,-0.2);
		\node (n5) at (10,-0.5){$5$};\draw[very thin, color=black!50] (10,0.2) -- (10,-0.2);
		\node (n6) at (12,-0.5){$6$};\draw[very thin, color=black!50] (12,0.2) -- (12,-0.2);
		\node (n7) at (14,-0.5){$7$};\draw[very thin, color=black!50] (14,0.2) -- (14,-0.2);
		\node (n8) at (16,-0.5){$8$};
		\draw[red, double, line width=0.1mm, ->] (0,0) to [out=40,in=140] (4,0);
		\draw[black, thin, ->] (0,0) to [out=-40,in=-140] (10,0);
		\draw[black, thin, ->] (4,0) to [out=-40,in=-140] (14,0);
		\draw[red, dotted, ->] (0,0) to [out=-50,in=-120] (2,0);
		\draw[red, dotted, ->] (4,0) to [out=-50,in=-120] (6,0);
		\draw[red, dotted, ->] (10,0) to [out=-50,in=-120] (12,0);
		\draw[red, dotted, ->] (14,0) to [out=-50,in=-120] (16,0);
		\draw[black, dashdotted, thin, ->] (0,0) to [out=40,in=140] (8,0);
		\draw[black, dashdotted, thin, ->] (2,0) to [out=40,in=140] (10,0);
		\draw[black, dashdotted, thin, ->] (4,0) to [out=40,in=140] (12,0);
		\draw[black, dashdotted, thin, ->] (6,0) to [out=40,in=140] (14,0);
		\draw[gray, dashed, thin, ->] (2,0) to [out=-40,in=-140] (16,0);
		\draw[gray, dashed, thin, ->] (4,0) to [out=-40,in=-140] (16,0);
		\draw[gray, dashed, thin, ->] (6,0) to [out=-40,in=-140] (16,0);
		\draw[gray, dashed, thin, ->] (8,0) to [out=-40,in=-140] (16,0);
		\draw[gray, dashed, thin, ->] (10,0) to [out=-40,in=-140] (16,0);
		\draw[gray, dashed, thin, ->] (12,0) to [out=-40,in=-140] (16,0);
		\draw[gray, dashed, thin, ->] (14,0) to [out=-40,in=-140] (16,0);
		\draw[red, double, line width=0.1mm, -](0,-1.5) to (1,-1.5);\node (j2) at (1.7,-1.5){\footnotesize{job 1}};
		\draw[black, thin, -](0,-1.8) to (1,-1.8);\node (j1) at (1.7,-1.8){\footnotesize{job 2}};
		\draw[red, dotted, -](0,-2.1) to (1,-2.1);\node (j3) at (1.7,-2.1){\footnotesize{job 3}};
		\draw[black, dashdotted, thin, -](0,-2.4) to (1,-2.4);\node (j4) at (1.7,-2.4){\footnotesize{job 4}};
		\draw[gray, dashed, thin, -](0,-2.7) to (1,-2.7);\node (j0) at (1.7,-2.7){\footnotesize{loss arcs}};		
		\end{tikzpicture}
	}\\
	\subfloat[Optimal solution value = 67]{
		\label{fig:Ex1Alg1c}
		\centering
		\begin{tikzpicture}[thick, scale=0.90, every node/.style={transform shape}, >=stealth', dot/.style = {draw, fill = white, circle, inner sep = 0pt, minimum size = 4pt}]
		\draw[very thin, color=black!50] (0,0) -- (16,0);
		\node (nT) at (16,0.5){T};\draw[color=red] (16,0.35) -- (16,-0.35);
		\node (n0) at (0,-0.5){$0$}; \draw[very thin, color=black!50] (0,0.2) -- (0,-0.2);\draw[color=black] (0,0.35) -- (0,-0.35);
		\node (n1) at (2,-0.5){$1$};\draw[very thin, color=black!50] (2,0.2) -- (2,-0.2);
		\node (n2) at (4,-0.5){$2$};\draw[very thin, color=black!50] (4,0.2) -- (4,-0.2);
		\node (n3) at (6,-0.5){$3$};\draw[very thin, color=black!50] (6,0.2) -- (6,-0.2);
		\node (n4) at (8,-0.5){$4$};\draw[very thin, color=black!50] (8,0.2) -- (8,-0.2);
		\node (n5) at (10,-0.5){$5$};\draw[very thin, color=black!50] (10,0.2) -- (10,-0.2);
		\node (n6) at (12,-0.5){$6$};\draw[very thin, color=black!50] (12,0.2) -- (12,-0.2);
		\node (n7) at (14,-0.5){$7$};\draw[very thin, color=black!50] (14,0.2) -- (14,-0.2);
		\node (n8) at (16,-0.5){$8$};
		\draw[red, double, line width=0.1mm, ->] (0,0) to [out=40,in=140] (4,0);
		\draw[black, thin, ->] (0,0) to [out=-40,in=-140] (10,0);
		\draw[red, dotted, ->] (4,0) to [out=-50,in=-120] (6,0);
		\draw[black, dashdotted, thin, ->] (6,0) to [out=40,in=140] (14,0);
		\draw[gray, dashed, thin, ->] (10,0) to [out=-40,in=-140] (16,0);
		\draw[gray, dashed, thin, ->] (14,0) to [out=-40,in=-140] (16,0);
		\draw[red, double, line width=0.1mm, -](0,-1.5) to (1,-1.5);\node (j2) at (1.7,-1.5){\footnotesize{job 1}};
		\draw[black, thin, -](0,-1.8) to (1,-1.8);\node (j1) at (1.7,-1.8){\footnotesize{job 2}};
		\draw[red, dotted, -](0,-2.1) to (1,-2.1);\node (j3) at (1.7,-2.1){\footnotesize{job 3}};
		\draw[black, dashdotted, thin, -](0,-2.4) to (1,-2.4);\node (j4) at (1.7,-2.4){\footnotesize{job 4}};
		\draw[gray, dashed, thin, -](0,-2.7) to (1,-2.7);\node (j0) at (1.7,-2.7){\footnotesize{loss arcs}};			
		\end{tikzpicture}
	}
	\caption{Example of AF formulation}
	\label{fig:Ex1Alg1}
\end{figure}

\rev{We now notice a relevant property of the AF formulation.		
	\begin{proposition}
		The AF formulation \eqref{FO:AF}--\eqref{constr4:AF} is equivalent to the TI formulation \eqref{FO:TI}--\eqref{constr3:TI}.
	\end{proposition}
	\begin{proof}	
		Our proof uses arguments similar to that adopted in \cite{ValeriodeCarvalho2002} for the proof of his Proposition 5.1 (which shows the equivalence of the AF and TI formulations for the cutting stock problem). To simplify the reasoning, we consider a basic AF formulation in which the vertices set $N$  contains all nodes from $0$ to $T$ (i.e., we do not apply the reduction due to the normal patterns) and the arc set contains all the job and loss arcs. Moreover,
		we substitute the `$\geq$' sign in \eqref{constr2:AF} with the `$=$' sign, without loosing optimal solutions because any solution selecting more than one arc for the same job can be improved by choosing only one of these arcs.
		
		Remind that the three indices of the $x$ variables in AF are introduced to simplify the writing of the model, but only two indices are necessary. Indeed, job arc $(q,r,j)$ is introduced only when $r = q+p_j$, and hence $x_{qrj}$ is set to one, in AF, if job  $j$ starts at time $q$, as $x_{jt}$  is set to one, in TI, when job $j$ starts at time $t$. Using this observation one can see that the objective functions \eqref{FO:AF} and \eqref{FO:TI} are equivalent.
		
		Both constraints \eqref{constr2:AF} in AF (with the `$=$' sign) and constraints \eqref{constr1:TI} in TI impose that a single starting time for each job $j$ is chosen.
		
		We conclude the proof by showing that the remaining constraints in AF and TI are equivalent, by using an unimodular transformation, similar to the one in \cite{ValeriodeCarvalho2002}.
		To better understand this transformation we refer to the example of Figure \ref{fig:Ex1Alg1}, having two machines, 
		four jobs and an upper bound $T=8$ for the completion of all jobs.
		In Figure \ref{fig:UT_ti} we report the constraint matrix of constraints \eqref{constr2:TI} in TI, while in Figure \ref{fig:UT_af} we report the constraint matrix of the flow conservation constraints \eqref{constr1:AF} in AF (note that in the TI matrix we added an empty row corresponding to $t=T$, for easier comparison with the AF matrix). In the TI matrix we have a column for each starting time of each job.
		In the AF formulation the first part of the matrix refers to the job arcs, while the last part report on the loss arcs used to model the empty space in the machines. In the first part of the two matrices there is a column by column  correspondence: variable $x_{qrj}$ in AF defines a possible starting time of job $j$ at time $q$, exactly as variable $x_{jq}$ do in TI.
		
		Consider the AF constraint matrix and, for each $t=0,\dots, T$, let us substitute the $t$-th flow conservation constraint with the sum of the first $t+1$ constraints \eqref{constr1:AF}, thus obtaining an equivalent constraint matrix. The resulting matrix is depicted in Figure \ref{fig:UT_afut}.
		
		One can note that the first part of the new constraint matrix is identical to the first part of the TI matrix. In the second part we have the same r.h.s. as in TI, but equal sign instead of `$\leq$' sign. However,  each constraint $t$ is completed with the sum of the first $t+1$ loss arcs:  $\sum_{q=0}^t x_{qT0}$, thus giving at least one slack variable available for each constraint.  It follows that the new matrix is equivalent to the TI matrix, which conclude the proof.
\end{proof}}							

\begin{figure}[htbp]
	\centering
	\scriptsize
	\setlength{\tabcolsep}{0.5mm}
	\subfloat[Time-indexed formulation]{
		\label{fig:UT_ti}
		\rev{
			\begin{tabular}{r|p{2.5mm}p{2.5mm}p{2.5mm}p{2.5mm}p{2.5mm}p{2.5mm}p{2.5mm}p{2.5mm}p{2.5mm}p{2.5mm}p{2.5mm}p{2.5mm}p{2.5mm}p{2.5mm}p{2.5mm}p{2.5mm}p{2.5mm}p{2.5mm}p{2.5mm}p{2.5mm}p{2.5mm}p{2.5mm}p{2.5mm}p{2.5mm}p{2.5mm}p{2.5mm}p{2.5mm}p{2.5mm}p{2.5mm}p{2.5mm}p{2.5mm}p{2.5mm}|p{2.5mm}p{2.5mm}p{2.5mm}p{2.5mm}p{2.5mm}p{2.5mm}p{2.5mm}p{2.5mm}p{2.5mm}p{2.5mm}}
				\toprule
				\multicolumn{1}{l|}{$t$}      & \multicolumn{1}{l}{\begin{sideways}$x_{10}$\end{sideways}} & \multicolumn{1}{l}{\begin{sideways}$x_{2 0}$\end{sideways}} & \multicolumn{1}{l}{\begin{sideways}$x_{30}$\end{sideways}} & \multicolumn{1}{l}{\begin{sideways}$x_{40}$\end{sideways}} & \multicolumn{1}{l}{\begin{sideways}$x_{11}$\end{sideways}} & \multicolumn{1}{l}{\begin{sideways}$x_{21}$\end{sideways}} & \multicolumn{1}{l}{\begin{sideways}$x_{31}$\end{sideways}} & \multicolumn{1}{l}{\begin{sideways}$x_{41}$\end{sideways}} & \multicolumn{1}{l}{\begin{sideways}$x_{12}$\end{sideways}} & \multicolumn{1}{l}{\begin{sideways}$x_{22}$\end{sideways}} & \multicolumn{1}{l}{\begin{sideways}$x_{32}$\end{sideways}} & \multicolumn{1}{l}{\begin{sideways}$x_{42}$\end{sideways}} & \multicolumn{1}{l}{\begin{sideways}$x_{13}$\end{sideways}} & \multicolumn{1}{l}{\begin{sideways}$x_{23}$\end{sideways}} & \multicolumn{1}{l}{\begin{sideways}$x_{33}$\end{sideways}} & \multicolumn{1}{l}{\begin{sideways}$x_{43}$\end{sideways}} & \multicolumn{1}{l}{\begin{sideways}$x_{14}$\end{sideways}} & \multicolumn{1}{l}{\begin{sideways}$x_{24}$\end{sideways}} & \multicolumn{1}{l}{\begin{sideways}$x_{34}$\end{sideways}} & \multicolumn{1}{l}{\begin{sideways}$x_{44}$\end{sideways}} & \multicolumn{1}{l}{\begin{sideways}$x_{15}$\end{sideways}} & \multicolumn{1}{l}{\begin{sideways}$x_{25}$\end{sideways}} & \multicolumn{1}{l}{\begin{sideways}$x_{35}$\end{sideways}} & \multicolumn{1}{l}{\begin{sideways}$x_{45}$\end{sideways}} & \multicolumn{1}{l}{\begin{sideways}$x_{16}$\end{sideways}} & \multicolumn{1}{l}{\begin{sideways}$x_{26}$\end{sideways}} & \multicolumn{1}{l}{\begin{sideways}$x_{36}$\end{sideways}} & \multicolumn{1}{l}{\begin{sideways}$x_{46}$\end{sideways}} & \multicolumn{1}{l}{\begin{sideways}$x_{17}$\end{sideways}} & \multicolumn{1}{l}{\begin{sideways}$x_{27}$\end{sideways}} & \multicolumn{1}{l}{\begin{sideways}$x_{37}$\end{sideways}} & \multicolumn{1}{l|}{\begin{sideways}$x_{47}$\end{sideways}} &       &       &       &       &       &       &       &       &       &  \multicolumn{1}{l}{\begin{sideways}\end{sideways}}  \\
				\midrule
				0     & \multicolumn{1}{r}{1}     & \multicolumn{1}{r}{1}     & \multicolumn{1}{r}{1}     & \multicolumn{1}{r}{1}     &       &       &       &       &       &       &       &       &       &       &       &       &       &       &       &       &       &       &       &       &       &       &       &       &       &       &       &       &       &       &       &       &       &       &       &       & $\leq$    & \multicolumn{1}{r}{2} \\
				1     & \multicolumn{1}{r}{1}     & \multicolumn{1}{r}{1}     &       & \multicolumn{1}{r}{1}     & \multicolumn{1}{r}{1}     & \multicolumn{1}{r}{1}     & \multicolumn{1}{r}{1}     & \multicolumn{1}{r}{1}     &       &       &       &       &       &       &       &       &       &       &       &       &       &       &       &       &       &       &       &       &       &       &       &       &       &       &       &       &       &       &       &       & $\leq$    & \multicolumn{1}{r}{2} \\
				2     &       & \multicolumn{1}{r}{1}     &       & \multicolumn{1}{r}{1}     & \multicolumn{1}{r}{1}     & \multicolumn{1}{r}{1}     &       & \multicolumn{1}{r}{1}     & \multicolumn{1}{r}{1}     & \multicolumn{1}{r}{1}     & \multicolumn{1}{r}{1}     & \multicolumn{1}{r}{1}     &       &       &       &       &       &       &       &       &       &       &       &       &       &       &       &       &       &       &       &       &       &       &       &       &       &       &       &       & $\leq$    & \multicolumn{1}{r}{2} \\
				3     &       & \multicolumn{1}{r}{1}     &       & \multicolumn{1}{r}{1}     &       & \multicolumn{1}{r}{1}     &       & \multicolumn{1}{r}{1}     & \multicolumn{1}{r}{1}     & \multicolumn{1}{r}{1}     &       & \multicolumn{1}{r}{1}     & \multicolumn{1}{r}{1}     & \multicolumn{1}{r}{1}     & \multicolumn{1}{r}{1}     & \multicolumn{1}{r}{1}     &       &       &       &       &       &       &       &       &       &       &       &       &       &       &       &       &       &       &       &       &       &       &       &       & $\leq$    & \multicolumn{1}{r}{2} \\
				4     &       & \multicolumn{1}{r}{1}     &       &       &       & \multicolumn{1}{r}{1}     &       & \multicolumn{1}{r}{1}     &       & \multicolumn{1}{r}{1}     &       & \multicolumn{1}{r}{1}     & \multicolumn{1}{r}{1}     & \multicolumn{1}{r}{1}     &       & \multicolumn{1}{r}{1}     & \multicolumn{1}{r}{1}     &       & \multicolumn{1}{r}{1}     & \multicolumn{1}{r}{1}     &       &       &       &       &       &       &       &       &       &       &       &       &       &       &       &       &       &       &       &       & $\leq$    & \multicolumn{1}{r}{2} \\
				5     &       &       &       &       &       & \multicolumn{1}{r}{1}     &       &       &       & \multicolumn{1}{r}{1}     &       & \multicolumn{1}{r}{1}     &       & \multicolumn{1}{r}{1}     &       & \multicolumn{1}{r}{1}     & \multicolumn{1}{r}{1}     &       &       & \multicolumn{1}{r}{1}     & \multicolumn{1}{r}{1}     &       & \multicolumn{1}{r}{1}     &       &       &       &       &       &       &       &       &       &       &       &       &       &       &       &       &       & $\leq$    & \multicolumn{1}{r}{2} \\
				6     &       &       &       &       &       &       &       &       &       & \multicolumn{1}{r}{1}     &       &       &       & \multicolumn{1}{r}{1}     &       & \multicolumn{1}{r}{1}     &       &       &       & \multicolumn{1}{r}{1}     & \multicolumn{1}{r}{1}     &       &       &       & \multicolumn{1}{r}{1}     &       & \multicolumn{1}{r}{1}     &       &       &       &       &       &       &       &       &       &       &       &       &       & $\leq$    & \multicolumn{1}{r}{2} \\
				7     &       &       &       &       &       &       &       &       &       &       &       &       &       & \multicolumn{1}{r}{1}     &       &       &       &       &       & \multicolumn{1}{r}{1}     &       &       &       &       & \multicolumn{1}{r}{1}     &       &       &       &       &       & \multicolumn{1}{r}{1}     &       &       &       &       &       &       &       &       &       & $\leq$    & \multicolumn{1}{r}{2} \\
				8     &       &       &       &       &       &       &       &       &       &       &       &       &       &       &       &       &       &       &       &       &       &       &       &       &       &       &       &       &       &       &       &       &       &       &       &       &       &       &       &       & $\leq$    & \multicolumn{1}{r}{0} \\
				\bottomrule
				\multicolumn{43}{l}{}\\	
			\end{tabular}
		}
	}\\
	\subfloat[Arc-flow formulation]{
		\label{fig:UT_af}
		\rev{
			\begin{tabular}{r|p{2.5mm}p{2.5mm}p{2.5mm}p{2.5mm}p{2.5mm}p{2.5mm}p{2.5mm}p{2.5mm}p{2.5mm}p{2.5mm}p{2.5mm}p{2.5mm}p{2.5mm}p{2.5mm}p{2.5mm}p{2.5mm}p{2.5mm}p{2.5mm}p{2.5mm}p{2.5mm}p{2.5mm}p{2.5mm}p{2.5mm}p{2.5mm}p{2.5mm}p{2.5mm}p{2.5mm}p{2.5mm}p{2.5mm}p{2.5mm}p{2.5mm}p{2.5mm}|p{2.5mm}p{2.5mm}p{2.5mm}p{2.5mm}p{2.5mm}p{2.5mm}p{2.5mm}p{2.5mm}p{2.5mm}p{2.5mm}}
				\toprule
				\multicolumn{1}{l|}{$q$}    & \multicolumn{1}{l}{\begin{sideways}$x_{021}$\end{sideways}} & \multicolumn{1}{l}{\begin{sideways}$x_{05 2}$\end{sideways}} & \multicolumn{1}{l}{\begin{sideways}$x_{013}$\end{sideways}} & \multicolumn{1}{l}{\begin{sideways}$x_{044}$\end{sideways}} & \multicolumn{1}{l}{\begin{sideways}$x_{131}$\end{sideways}} & \multicolumn{1}{l}{\begin{sideways}$x_{162}$\end{sideways}} & \multicolumn{1}{l}{\begin{sideways}$x_{123}$\end{sideways}} & \multicolumn{1}{l}{\begin{sideways}$x_{154}$\end{sideways}} & \multicolumn{1}{l}{\begin{sideways}$x_{241}$\end{sideways}} & \multicolumn{1}{l}{\begin{sideways}$x_{272}$\end{sideways}} & \multicolumn{1}{l}{\begin{sideways}$x_{233}$\end{sideways}} & \multicolumn{1}{l}{\begin{sideways}$x_{264}$\end{sideways}} & \multicolumn{1}{l}{\begin{sideways}$x_{351}$\end{sideways}} & \multicolumn{1}{l}{\begin{sideways}$x_{382}$\end{sideways}} & \multicolumn{1}{l}{\begin{sideways}$x_{343}$\end{sideways}} & \multicolumn{1}{l}{\begin{sideways}$x_{374}$\end{sideways}} & \multicolumn{1}{l}{\begin{sideways}$x_{461}$\end{sideways}} & \multicolumn{1}{l}{\begin{sideways}$x_{492}$\end{sideways}} & \multicolumn{1}{l}{\begin{sideways}$x_{453}$\end{sideways}} & \multicolumn{1}{l}{\begin{sideways}$x_{484}$\end{sideways}} & \multicolumn{1}{l}{\begin{sideways}$x_{571}$\end{sideways}} & \multicolumn{1}{l}{\begin{sideways}$x_{5102}$\end{sideways}} & \multicolumn{1}{l}{\begin{sideways}$x_{563}$\end{sideways}} & \multicolumn{1}{l}{\begin{sideways}$x_{594}$\end{sideways}} & \multicolumn{1}{l}{\begin{sideways}$x_{681}$\end{sideways}} & \multicolumn{1}{l}{\begin{sideways}$x_{6112}$\end{sideways}} & \multicolumn{1}{l}{\begin{sideways}$x_{673}$\end{sideways}} & \multicolumn{1}{l}{\begin{sideways}$x_{6104}$\end{sideways}} & \multicolumn{1}{l}{\begin{sideways}$x_{791}$\end{sideways}} & \multicolumn{1}{l}{\begin{sideways}$x_{7122}$\end{sideways}} & \multicolumn{1}{l}{\begin{sideways}$x_{783}$\end{sideways}} & \multicolumn{1}{l|}{\begin{sideways}$x_{7114}$\end{sideways}} & \multicolumn{1}{l}{\begin{sideways}$x_{080}$\end{sideways}} & \multicolumn{1}{l}{\begin{sideways}$x_{180}$\end{sideways}} & \multicolumn{1}{l}{\begin{sideways}$x_{280}$\end{sideways}} & \multicolumn{1}{l}{\begin{sideways}$x_{380}$\end{sideways}} & \multicolumn{1}{l}{\begin{sideways}$x_{480}$\end{sideways}} & \multicolumn{1}{l}{\begin{sideways}$x_{580}$\end{sideways}} & \multicolumn{1}{l}{\begin{sideways}$x_{680}$\end{sideways}} & \multicolumn{1}{l}{\begin{sideways}$x_{780}$\end{sideways}} &       &  \multicolumn{1}{l}{\begin{sideways}\end{sideways}}  \\
				\midrule
				0     & \multicolumn{1}{r}{1}     & \multicolumn{1}{r}{1}     & \multicolumn{1}{r}{1}     & \multicolumn{1}{r}{1}     &       &       &       &       &       &       &       &       &       &       &       &       &       &       &       &       &       &       &       &       &       &       &       &       &       &       &       &       & \multicolumn{1}{r}{1}     &       &       &       &       &       &       &       & =     & \multicolumn{1}{r}{2} \\
				1     &       &       & \multicolumn{1}{r}{-1}    &       & \multicolumn{1}{r}{1}     & \multicolumn{1}{r}{1}     & \multicolumn{1}{r}{1}     & \multicolumn{1}{r}{1}     &       &       &       &       &       &       &       &       &       &       &       &       &       &       &       &       &       &       &       &       &       &       &       &       &       & \multicolumn{1}{r}{1}     &       &       &       &       &       &       & =     & \multicolumn{1}{r}{0} \\
				2     & \multicolumn{1}{r}{-1}    &       &       &       &       &       & \multicolumn{1}{r}{-1}    &       & \multicolumn{1}{r}{1}     & \multicolumn{1}{r}{1}     & \multicolumn{1}{r}{1}     & \multicolumn{1}{r}{1}     &       &       &       &       &       &       &       &       &       &       &       &       &       &       &       &       &       &       &       &       &       &       & \multicolumn{1}{r}{1}     &       &       &       &       &       & =     & \multicolumn{1}{r}{0} \\
				3     &       &       &       &       & \multicolumn{1}{r}{-1}    &       &       &       &       &       & \multicolumn{1}{r}{-1}    &       & \multicolumn{1}{r}{1}     & \multicolumn{1}{r}{1}     & \multicolumn{1}{r}{1}     & \multicolumn{1}{r}{1}     &       &       &       &       &       &       &       &       &       &       &       &       &       &       &       &       &       &       &       & \multicolumn{1}{r}{1}     &       &       &       &       & =     & \multicolumn{1}{r}{0} \\
				4     &       &       &       & \multicolumn{1}{r}{-1}    &       &       &       &       & \multicolumn{1}{r}{-1}    &       &       &       &       &       & \multicolumn{1}{r}{-1}    &       & \multicolumn{1}{r}{1}     &       & \multicolumn{1}{r}{1}     & \multicolumn{1}{r}{1}     &       &       &       &       &       &       &       &       &       &       &       &       &       &       &       &       & \multicolumn{1}{r}{1}     &       &       &       & =     & \multicolumn{1}{r}{0} \\
				5     &       & \multicolumn{1}{r}{-1}    &       &       &       &       &       & \multicolumn{1}{r}{-1}    &       &       &       &       & \multicolumn{1}{r}{-1}    &       &       &       &       &       & \multicolumn{1}{r}{-1}    &       & \multicolumn{1}{r}{1}     &       & \multicolumn{1}{r}{1}     &       &       &       &       &       &       &       &       &       &       &       &       &       &       & \multicolumn{1}{r}{1}     &       &       & =     & \multicolumn{1}{r}{0} \\
				6     &       &       &       &       &       & \multicolumn{1}{r}{-1}    &       &       &       &       &       & \multicolumn{1}{r}{-1}    &       &       &       &       & \multicolumn{1}{r}{-1}    &       &       &       &       &       & \multicolumn{1}{r}{-1}    &       & \multicolumn{1}{r}{1}     &       & \multicolumn{1}{r}{1}     &       &       &       &       &       &       &       &       &       &       &       & \multicolumn{1}{r}{1}     &       & =     & \multicolumn{1}{r}{0} \\
				7     &       &       &       &       &       &       &       &       &       & \multicolumn{1}{r}{-1}    &       &       &       &       &       & \multicolumn{1}{r}{-1}    &       &       &       &       & \multicolumn{1}{r}{-1}    &       &       &       &       &       & \multicolumn{1}{r}{-1}    &       &       &       & \multicolumn{1}{r}{1}     &       &       &       &       &       &       &       &       & \multicolumn{1}{r}{1}     & =     & \multicolumn{1}{r}{0} \\
				8     &       &       &       &       &       &       &       &       &       &       &       &       &       & \multicolumn{1}{r}{-1}    &       &       &       &       &       & \multicolumn{1}{r}{-1}    &       &       &       &       & \multicolumn{1}{r}{-1}    &       &       &       &       &       & \multicolumn{1}{r}{-1}    &       & \multicolumn{1}{r}{-1}    & \multicolumn{1}{r}{-1}    & \multicolumn{1}{r}{-1}    & \multicolumn{1}{r}{-1}    & \multicolumn{1}{r}{-1}    & \multicolumn{1}{r}{-1}    & \multicolumn{1}{r}{-1}    & \multicolumn{1}{r}{-1}    & =     & \multicolumn{1}{r}{-2} \\
				\bottomrule
				\multicolumn{43}{l}{}\\
			\end{tabular}%
		}
	}\\	
	\subfloat[Arc-flow formulation after unimodular transformation]{
		\label{fig:UT_afut}
		\rev{
			\begin{tabular}{r|p{2.5mm}p{2.5mm}p{2.5mm}p{2.5mm}p{2.5mm}p{2.5mm}p{2.5mm}p{2.5mm}p{2.5mm}p{2.5mm}p{2.5mm}p{2.5mm}p{2.5mm}p{2.5mm}p{2.5mm}p{2.5mm}p{2.5mm}p{2.5mm}p{2.5mm}p{2.5mm}p{2.5mm}p{2.5mm}p{2.5mm}p{2.5mm}p{2.5mm}p{2.5mm}p{2.5mm}p{2.5mm}p{2.5mm}p{2.5mm}p{2.5mm}p{2.5mm}|p{2.5mm}p{2.5mm}p{2.5mm}p{2.5mm}p{2.5mm}p{2.5mm}p{2.5mm}p{2.5mm}p{2.5mm}p{2.5mm}}
				\toprule
				\multicolumn{1}{l|}{$q$}     & \multicolumn{1}{l}{\begin{sideways}$x_{021}$\end{sideways}} & \multicolumn{1}{l}{\begin{sideways}$x_{05 2}$\end{sideways}} & \multicolumn{1}{l}{\begin{sideways}$x_{013}$\end{sideways}} & \multicolumn{1}{l}{\begin{sideways}$x_{044}$\end{sideways}} & \multicolumn{1}{l}{\begin{sideways}$x_{131}$\end{sideways}} & \multicolumn{1}{l}{\begin{sideways}$x_{162}$\end{sideways}} & \multicolumn{1}{l}{\begin{sideways}$x_{123}$\end{sideways}} & \multicolumn{1}{l}{\begin{sideways}$x_{154}$\end{sideways}} & \multicolumn{1}{l}{\begin{sideways}$x_{241}$\end{sideways}} & \multicolumn{1}{l}{\begin{sideways}$x_{272}$\end{sideways}} & \multicolumn{1}{l}{\begin{sideways}$x_{233}$\end{sideways}} & \multicolumn{1}{l}{\begin{sideways}$x_{264}$\end{sideways}} & \multicolumn{1}{l}{\begin{sideways}$x_{351}$\end{sideways}} & \multicolumn{1}{l}{\begin{sideways}$x_{382}$\end{sideways}} & \multicolumn{1}{l}{\begin{sideways}$x_{343}$\end{sideways}} & \multicolumn{1}{l}{\begin{sideways}$x_{374}$\end{sideways}} & \multicolumn{1}{l}{\begin{sideways}$x_{461}$\end{sideways}} & \multicolumn{1}{l}{\begin{sideways}$x_{492}$\end{sideways}} & \multicolumn{1}{l}{\begin{sideways}$x_{453}$\end{sideways}} & \multicolumn{1}{l}{\begin{sideways}$x_{484}$\end{sideways}} & \multicolumn{1}{l}{\begin{sideways}$x_{571}$\end{sideways}} & \multicolumn{1}{l}{\begin{sideways}$x_{5102}$\end{sideways}} & \multicolumn{1}{l}{\begin{sideways}$x_{563}$\end{sideways}} & \multicolumn{1}{l}{\begin{sideways}$x_{594}$\end{sideways}} & \multicolumn{1}{l}{\begin{sideways}$x_{681}$\end{sideways}} & \multicolumn{1}{l}{\begin{sideways}$x_{6112}$\end{sideways}} & \multicolumn{1}{l}{\begin{sideways}$x_{673}$\end{sideways}} & \multicolumn{1}{l}{\begin{sideways}$x_{6104}$\end{sideways}} & \multicolumn{1}{l}{\begin{sideways}$x_{791}$\end{sideways}} & \multicolumn{1}{l}{\begin{sideways}$x_{7122}$\end{sideways}} & \multicolumn{1}{l}{\begin{sideways}$x_{783}$\end{sideways}} & \multicolumn{1}{l|}{\begin{sideways}$x_{7114}$\end{sideways}} & \multicolumn{1}{l}{\begin{sideways}$x_{080}$\end{sideways}} & \multicolumn{1}{l}{\begin{sideways}$x_{180}$\end{sideways}} & \multicolumn{1}{l}{\begin{sideways}$x_{280}$\end{sideways}} & \multicolumn{1}{l}{\begin{sideways}$x_{380}$\end{sideways}} & \multicolumn{1}{l}{\begin{sideways}$x_{480}$\end{sideways}} & \multicolumn{1}{l}{\begin{sideways}$x_{580}$\end{sideways}} & \multicolumn{1}{l}{\begin{sideways}$x_{680}$\end{sideways}} & \multicolumn{1}{l}{\begin{sideways}$x_{780}$\end{sideways}} &       &  \multicolumn{1}{l}{\begin{sideways}\end{sideways}}  \\
				\midrule
				0     & \multicolumn{1}{r}{1}     & \multicolumn{1}{r}{1}     & \multicolumn{1}{r}{1}     & \multicolumn{1}{r}{1}     &       &       &       &       &       &       &       &       &       &       &       &       &       &       &       &       &       &       &       &       &       &       &       &       &       &       &       &       & \multicolumn{1}{r}{1}     &       &       &       &       &       &       &       & =     & \multicolumn{1}{r}{2} \\
				1     & \multicolumn{1}{r}{1}     & \multicolumn{1}{r}{1}     &       & \multicolumn{1}{r}{1}     & \multicolumn{1}{r}{1}     & \multicolumn{1}{r}{1}     & \multicolumn{1}{r}{1}     & \multicolumn{1}{r}{1}     &       &       &       &       &       &       &       &       &       &       &       &       &       &       &       &       &       &       &       &       &       &       &       &       & \multicolumn{1}{r}{1}     & \multicolumn{1}{r}{1}     &       &       &       &       &       &       & =     & \multicolumn{1}{r}{2} \\
				2     &       & \multicolumn{1}{r}{1}     &       & \multicolumn{1}{r}{1}     & \multicolumn{1}{r}{1}     & \multicolumn{1}{r}{1}     &       & \multicolumn{1}{r}{1}     & \multicolumn{1}{r}{1}     & \multicolumn{1}{r}{1}     & \multicolumn{1}{r}{1}     & \multicolumn{1}{r}{1}     &       &       &       &       &       &       &       &       &       &       &       &       &       &       &       &       &       &       &       &       & \multicolumn{1}{r}{1}     & \multicolumn{1}{r}{1}     & \multicolumn{1}{r}{1}     &       &       &       &       &       & =     & \multicolumn{1}{r}{2} \\
				3     &       & \multicolumn{1}{r}{1}     &       & \multicolumn{1}{r}{1}     &       & \multicolumn{1}{r}{1}     &       & \multicolumn{1}{r}{1}     & \multicolumn{1}{r}{1}     & \multicolumn{1}{r}{1}     &       & \multicolumn{1}{r}{1}     & \multicolumn{1}{r}{1}     & \multicolumn{1}{r}{1}     & \multicolumn{1}{r}{1}     & \multicolumn{1}{r}{1}     &       &       &       &       &       &       &       &       &       &       &       &       &       &       &       &       & \multicolumn{1}{r}{1}     & \multicolumn{1}{r}{1}     & \multicolumn{1}{r}{1}     & \multicolumn{1}{r}{1}     &       &       &       &       & =     & \multicolumn{1}{r}{2} \\
				4     &       & \multicolumn{1}{r}{1}     &       &       &       & \multicolumn{1}{r}{1}     &       & \multicolumn{1}{r}{1}     &       & \multicolumn{1}{r}{1}     &       & \multicolumn{1}{r}{1}     & \multicolumn{1}{r}{1}     & \multicolumn{1}{r}{1}     &       & \multicolumn{1}{r}{1}     & \multicolumn{1}{r}{1}     &       & \multicolumn{1}{r}{1}     & \multicolumn{1}{r}{1}     &       &       &       &       &       &       &       &       &       &       &       &       & \multicolumn{1}{r}{1}     & \multicolumn{1}{r}{1}     & \multicolumn{1}{r}{1}     & \multicolumn{1}{r}{1}     & \multicolumn{1}{r}{1}     &       &       &       & =     & \multicolumn{1}{r}{2} \\
				5     &       &       &       &       &       & \multicolumn{1}{r}{1}     &       &       &       & \multicolumn{1}{r}{1}     &       & \multicolumn{1}{r}{1}     &       & \multicolumn{1}{r}{1}     &       & \multicolumn{1}{r}{1}     & \multicolumn{1}{r}{1}     &       &       & \multicolumn{1}{r}{1}     & \multicolumn{1}{r}{1}     &       & \multicolumn{1}{r}{1}     &       &       &       &       &       &       &       &       &       & \multicolumn{1}{r}{1}     & \multicolumn{1}{r}{1}     & \multicolumn{1}{r}{1}     & \multicolumn{1}{r}{1}     & \multicolumn{1}{r}{1}     & \multicolumn{1}{r}{1}     &       &       & =     & \multicolumn{1}{r}{2} \\
				6     &       &       &       &       &       &       &       &       &       & \multicolumn{1}{r}{1}     &       &       &       & \multicolumn{1}{r}{1}     &       & \multicolumn{1}{r}{1}     &       &       &       & \multicolumn{1}{r}{1}     & \multicolumn{1}{r}{1}     &       &       &       & \multicolumn{1}{r}{1}     &       & \multicolumn{1}{r}{1}     &       &       &       &       &       & \multicolumn{1}{r}{1}     & \multicolumn{1}{r}{1}     & \multicolumn{1}{r}{1}     & \multicolumn{1}{r}{1}     & \multicolumn{1}{r}{1}     & \multicolumn{1}{r}{1}     & \multicolumn{1}{r}{1}     &       & =     & \multicolumn{1}{r}{2} \\
				7     &       &       &       &       &       &       &       &       &       &       &       &       &       & \multicolumn{1}{r}{1}     &       &       &       &       &       & \multicolumn{1}{r}{1}     &       &       &       &       & \multicolumn{1}{r}{1}     &       &       &       &       &       & \multicolumn{1}{r}{1}     &       & \multicolumn{1}{r}{1}     & \multicolumn{1}{r}{1}     & \multicolumn{1}{r}{1}     & \multicolumn{1}{r}{1}     & \multicolumn{1}{r}{1}     & \multicolumn{1}{r}{1}     & \multicolumn{1}{r}{1}     & \multicolumn{1}{r}{1}     & =     & \multicolumn{1}{r}{2} \\
				8     &       &       &       &       &       &       &       &       &       &       &       &       &       &       &       &       &       &       &       &       &       &       &       &       &       &       &       &       &       &       &       &       &       &       &       &       &       &       &       &       & =     & \multicolumn{1}{r}{0} \\
				\bottomrule
			\end{tabular}
		}
	}
	\caption{\rev{Illustration of the equivalence between TI and AF on the example of Figure \ref{fig:Ex1Alg1}}}
	\label{fig:UT}
\end{figure}
	
\subsection{Enhanced arc-flow formulation}\label{sec:method}
In this section, we show how to further reduce the number of variables and constraints required by the AF formulation \eqref{FO:AF}--\eqref{constr4:AF}, improving its computational behavior while preserving optimality.

The size of the AF formulation linearly depends from the horizon $T$, thus, a proper time horizon estimation is necessary. To this aim, we notice that \cite{vandenAkkeretal1999} considered the properties of an optimal schedule, originally developed by \cite{ElmaghrabyPark1974}, and remarked that there exists at least an optimal solution for which ``the last job on any machine is completed between time $H_{\min} = \frac{1}{m}\sum_{j \in J} p_j - \frac{(m-1)}{m}p_{\max}$ and $H_{\max} = \frac{1}{m}\sum_{j \in J} p_j + \frac{(m-1)}{m}p_{\max}$", where $p_{\max} = \max\nolimits_{j\in J} p_j$. On the basis of this statement, we can use $H_{\max}$ to set $T$ as
\begin{equation}\label{eq:T}
T = \left\lfloor\frac{1}{m}\sum\nolimits_{j\in J} p_j + \frac{(m-1)}{m} p_{\max} \right\rfloor\\
\end{equation}
We can then use the value of $H_{\min}$ to limit the number of loss arcs. As the last job on any machine is completed at or after $H_{\min}$, we  can create only loss arcs starting from vertices $q \in N$ with $q \geq \lceil H_{\min} \rceil$. We further increase this bound by considering Property 1 in \cite{AzizogluandKirca19999b}, thus obtaining
\begin{equation} \label{eq:Tbar_b}
T' = \left\lceil\frac{1}{m}\sum\nolimits_{j\in J} p_j - \frac{ \sum_{k = 1}^{m-1}\bar{p}_k }{m} \right\rceil\\
\end{equation}
where $\bar{p}$ is an array containing the processing times of all jobs $j \in J$ in non-increasing order. Since $\sum_{k = 1}^{m-1}\bar{p}_k \leq (m-1)p_{\max}$ holds, the  value in \eqref{eq:Tbar_b} is not less than $H_{\min}$.

Our next enhancement relies on the creation of a so-called \emph{time window} $[a_j, b_j]$ for each job $j \in J$.
By using once more the properties of an optimal schedule in \cite{vandenAkkeretal1999}, we derive an earliest possible start time $a_j$ and a latest possible start time $b_j$ that guarantee the existence of an optimal solution. The values of $a_j$ and $b_j$ are based on the property that, if $w_j \geq w_k$ and $p_j \leq p_k$ for a certain pair of jobs $j$ and $k$, then there exists an optimal solution in which $j$ starts not later than $k$.
The time windows are computed as follows. For each $j \in J$, we first define $\mathcal{P}_j = \{ k \in J: k < j, w_k \geq w_j, p_k \leq p_j \}$ and $\mathcal{L}_j = \{ k \in J: k > j, w_k \leq w_j, p_k \geq p_j \}$.
Following the aforementioned property, there exists an optimal schedule in which all jobs in $\mathcal{P}_j$ start no later than $j$, so, if $\mathcal{P}_j$ contains at least $m$ elements, one may conclude that at least $|\mathcal{P}_j|-m+1$ jobs in $\mathcal{P}_j$ are finished before $j$ starts being processed. Consequently, if $|\mathcal{P}_j| < m$ we set $a_j = 0$, otherwise we set $a_j = \lceil \rho_j/m \rceil$, where $\rho_j$ is the sum of the $|\mathcal{P}_j|-m+1$ smallest processing times. In an analogous mode, one can note that there is an optimal solution in which $j$ starts no later than the jobs in $\mathcal{L}_j$. Consequently, for each job $j \in J$ a maximum starting time can be set as $b_j = T - \left\lceil \left(\sum_{k \in \mathcal{L}_j} p_k + p_j \right)/m \right\rceil$. In addition, if $\mathcal{L}_j = \emptyset$ then, as stated in \cite{BelouadahandPotts1994}, one can set ${b}_j = \lceil \left(\sum_{k \in J} p_k - p_j\right)/m  \rceil$.

The next procedure that we propose attempts to reduce the number of arcs by grouping together identical jobs. To this aim, we merge together all jobs $j \in J$ having identical $p_j$ and $w_j$ values into {\em job types}. Let $J^{\prime} = \{1, 2, \dots,n^{\prime}\}$ be the resulting set of job types, and $d_j$ be the number of jobs contained in each job type. With respect to the original AF formulation, this change involve creating a different set of arcs $A^{\prime}$ and replacing the original binary variables with integer variables, as shown next. This allows to reduce consistently the number of symmetries in the model.
The time windows for each job type $j \in J^{\prime}$ are simply obtained by setting $a_j = \min \{a_k: k \in J, p_k = p_j, w_k = w_j\}$ and $b_j = \max \{b_k: k \in J, p_k = p_j, w_k = w_j\}$.
We then create only arcs that start in a time $q \in [a_j, b_j]$, for each job type $j \in J^{\prime}$.

Our {\em enhanced arc flow formulation} (EAF) is then:
\begin{align}
(\mbox{EAF}) \quad \min \sum_{(q,r,j) \in A'} w_j q x_{qrj} + \sum_{j \in J'} w_j p_j \label{FO:EAF}\\
\sum_{(q,r,j) \in \delta'^{+}(q)} x_{qrj} - \sum_{(p,q,j) \in \delta'^{-}(q)} x_{pqj} = \left\{
\begin{array}{l l}
m, & \quad \text{ if $q = 0$}\\
-m, & \quad \text{ if $q = T$}\\
0, & \quad \text{otherwise}
\end{array}\right. & &&  q \in N'\label{constr1:EAF}\\
\sum_{(q,r,j) \in A'} x_{qrj} \geq d_j & &&  j \in J' \label{constr2:EAF}\\
x_{qrj} \in \{0,\dots,d_j\} & &&  (q,r,j) \in A' \setminus A'_0\label{constr3:EAF}\\
0 \leq x_{qT0} \leq m & && { (q,T,0) \in A'_0} \label{constr4:EAF}
\end{align}

The EAF model \eqref{FO:EAF}--\eqref{constr4:EAF} is based on a reduced multigraph $G'=(N', A')$, in which both sets of nodes and arcs are obtained by applying the above reductions criteria from the original graph $G$ used for AF. The EAF model considers the set of job types $J'$ instead of that of jobs $J$ in AF, and consequently adopts an integer variable $x_{qrj}$ giving the \rev{number of jobs of type} $j$ that are scheduled from $q$ to $r=q+p_j$. Each variable of this \rev{type might take a value} at {most} $d_j$, as stated in constraints \eqref{constr3:EAF}. \rev{Constraints} \eqref{constr1:EAF} impose flow conservation on the $m$ paths, and constraints \eqref{constr2:EAF} impose demand to be satisfied.

The way in which the EAF multi-graph $G'$ is built is shown in Algorithm \ref{alg:PandA_2}, which updates the previous Algorithm \ref{alg:PandA} used for AF.
The procedure initializes the sets $N'$ and $A'$ of nodes and arcs, respectively, to the empty set. It then considers at steps 4--9 one job type $j$ at a time, according to the WSPT rule, and creates the $A^{\prime}_j$ sets keeping into account that each job type $j$ contains $d_j$ identical jobs and should start at a $q \in [a_j,b_j]$.
\begin{algorithm}[htb]
	\caption{Construction of the EAF multi-graph} \label{alg:PandA_2}
	\begin{algorithmic}[1]
		\Procedure {CreatePatterns\_and\_Arcs}{$T$}
		\State {\bf{initialize}} $P[0 \dots T] \gets 0;$ {\footnotesize {\color{gray} \Comment $P$: array of size $T+1$}}
		\State {\bf{initialize}} $N' \gets \emptyset; A'[0 \dots n] \gets \emptyset;$ {\footnotesize {\color{gray} \Comment $N'$: set of {vertices}; $A'$: set of arcs}}
		\State $P[0] \gets 1;$
		\For {$j \in J'$ according to the WSPT rule}
		\For {$t \gets b_j$ down to $a_j$}
		\If {$P[t] = 1$}
		\For {$q \gets 1$ to $d_j$}
		\If { {$t + q p_j \leq b_j$} } { $P[t + q p_j] \gets 1; \rev{A'[j] \gets A'[j]} \cup \{(t, t + q p_j,j)\}$};  \EndIf
		\EndFor
		\EndIf
		\EndFor
		\EndFor
		\For {$t \gets 0$ to $T$}
		\If {$P[t] = 1$}
		\State {$N' \gets N' \cup \{t\}$};
		\If {$T' \leq t { < T} $} {$\rev{A'[0] \gets A'[0]} \cup \{(t,T,0)\}$;} {\footnotesize {\color{gray} \Comment \rev{$A'[0]$}: set of loss arcs}} \EndIf					
		\EndIf
		\EndFor
		\State ${A^{\prime} \gets \cup_{j \in J+} \rev{A^{\prime}[j]}}$
		\State \bf{return} ${N^{\prime},A^{\prime} }$
		\EndProcedure
	\end{algorithmic}
\end{algorithm}

\section{Computational experiments} \label{sec:results}

The discussed models have been coded in C++ and solved using Gurobi Optimizer 7.0. The experiments were performed by using a single thread on a PC equipped with an Intel Xeon E5530 $2.40$ GHz quad-core processor and $20$GB of RAM, running under Ubuntu $14.04.5$ LTS. We first discuss the \rev{benchmark instances used for the experiments}, then present an upper bounding procedure devised to speed up the convergence of the models, and finally we present an extensive computational evaluation.

\subsection{Benchmark instances}\label{subsec:instances}

In our experiments, we considered \rev{two benchmark sets of instances.}

\rev{The first set is derived from the one} proposed by \cite{BulbulandSen2017} for the $R||\sum w_jC_j$. Their set is made by instances with $n \in \{30, 100, 400, 1000\}$ and $m \in \{2,4,6,8,16,30\}$. Processing times $p_{j}^{k}$ (i.e., processing time of job $j$ in machine $k$) were drawn according to a uniform distribution $U[1,p_{\max}]$, where $p_{\max} \in \{20, 100\}$, and penalty weights $w_j$ were created using a uniform distribution $U[1,20]$. For each combination of $(n,m,p_{\max})$, except when $n =30$ and $m \in \{16, 30\}$, $10$ instances were created, resulting in a set of $440$ instances which is now available at \url{ http://people.sabanciuniv.edu/bulbul/papers/Bulbul_Sen_Rm_TWCT_data-results_JoS_2016.rar}.  We adapted these instances to the $P||\sum w_jC_j$ by imposing the processing time of each job $j$ to be equal to its processing time on the first unrelated machine in the $R||\sum w_jC_j$ (i.e., $p_j = p_j^1$ $\forall j$).
To better evaluate the performance of the models on large instances, we used the procedure adopted by \cite{BulbulandSen2017} {to create a new additional set with $n = 700$, obtaining in this way a total of 560 instances.}

\rev{The second set has been proposed by \cite{KowalczykandLeus2018} and consists of $2400$ instances with $n \in \{20, 50, 100, 150\}$ and $m \in \{3, 5, 8, 10, 12\}$. The instances are divided into six different classes according to the distribution of processing times and weights.  For each class and each combination of $(n,m)$, $20$ instances were created.}
	
\subsection{Upper bound by iterated local search}\label{subsec:UB}
To start the model with a valid upper bound, we developed a modified version of the {\emph{iterated local search}} (ILS) based metaheuristic of \cite{KramerS2015}. The original method consists of a multi-start ILS for general earliness-tardiness scheduling problems on unrelated machines, and incorporates special structures to reduce the complexity for exploring the neighborhoods.

In general words, the ILS  by \cite{KramerS2015} is composed by constructive, local search and perturbation phases. 			
We modified the construction and local search phases to take into account that in the $P||\sum w_jC_j$ all machine schedules follow the WSPT rule. That resulted in a speed up of the algorithm. The initial solutions, which are obtained either randomly or by a {{greedy randomized adaptive search procedure}} \rev{(GRASP)}, are now sorted according to the WSPT rule on each machine.
\rev{The two main differences with respect to \cite{KramerS2015} regard the GRASP construction procedure and the local search phase. For the GRASP, we initially sort jobs according to the WSPT rule instead of performing a random sorting. In addition, at the end of the procedure we consider each machine in turn and sort the jobs that have been assigned to it by using once more the WSPT rule.}

\rev{Regarding the local search, \cite{KramerS2015} employed a randomized variable neighborhood descent (RVND) procedure (see \citealt{MladenovicHansen1997}). In our modified version, the intra-machine neighborhood structures have been replaced by a simple WSPT sorting procedure, which is invoked at the end of each inter-machine neighborhood search, as depicted in Algorithm \ref{alg:RVND}.}

\rev{The ILS was  adopted to provide an initial feasible solution for all our methods below as follows. It was not executed for  small-size instances having $n \leq 100$ jobs. It was instead executed for 100 seconds for medium-size instances with $100 < n < 400$ jobs, and for 300 seconds for large-size instances with $n \geq 400$ jobs.}

\begin{algorithm}[ht]
	\caption{RVND}
	\label{alg:RVND}
	\begin{algorithmic}[1]
		\Procedure {RVND}{$\pi$} {\footnotesize {\color{gray} \Comment $\pi$ is the input solution}}
		\State {\bf{initialize}} $L = L_0;$ \label{alg:line_listRVND} {\footnotesize {\color{gray} \Comment $L_0$: list containing all inter-machines neighborhood structures}}
		\While {$L \neq \emptyset$}
		\State {select a neighborhood $N \in L$ at random;}
		\State {find $\pi'\in N$, the best neighbor solution of $\pi$;}
		\For {$k \gets 1$ to $m$} {sort jobs in $\pi'[k]$} according to WSPT; \label{alg:line_sortRVND} \EndFor
		\If{$f(\pi') < f(\pi)$}  {\footnotesize {\color{gray} \Comment $f(\pi)$ represents the cost of solution $\pi$}}
		\State {$\pi'\gets \pi'$;}
		\State {\bf{reinitialize} $L;$} 
		\Else
		\State {$L \gets L \setminus \{N\}$;}{\footnotesize {\color{gray} \Comment remove $N$ from $L$}}
		\EndIf
		\EndWhile
		\State \bf{return} $\pi$
		\EndProcedure
	\end{algorithmic}
\end{algorithm}

\subsection{Computational results on benchmark set 1}\label{subsec:experiments}
In Tables \ref{tab:res_p20} and \ref{tab:res_p100} we compare the performance of formulations \rev{CIQP (model \eqref{FO:CQP}--\eqref{constr3:CQP}),  PTI (model \eqref{FO:PTI}--\eqref{constr5:PTI}), TI (model \eqref{FO:TI}--\eqref{constr3:TI}), SC  (model \eqref{FO:SC}--\eqref{constr3:SC}), AF (model \eqref{FO:AF}--\eqref{constr4:AF}) and EAF (model \eqref{FO:EAF}--\eqref{constr4:EAF})}. To solve SC, we reimplemented the B\&P by \cite{vandenAkkeretal1999}. \rev{For instances with more than 100 jobs, each} method received as initial cutoff the upper bound \rev{produced by the ILS algorithm that we developed  using the time limit detailed at the end of} Section \ref{subsec:UB}.

Table \ref{tab:res_p20} summarizes the results that we obtained for the instances with $p_{\max} = 20$, whereas Table \ref{tab:res_p100} focuses on the case where $p_{\max} = 100$. For each group of $10$ instances defined by the couple $(n,m)$ and for each attempted method, we report the number of instances for which at least the root node of the model was solved, {\#root} (not reported for CIPQ and PTI) \rev{and} the number of optimal solutions found, {\#opt}.
\rev{In columns t(s) we report the average execution time in seconds for the 10 instances in the line. If for some of these instances either time or memory limit has been reached, then we consider the entire time limit in the computation of the average t(s) value. Note that we directly write \emph{t.lim}, respectively \emph{m.lim}, when the time limit, respectively memory limit, has been reached on all the 10 instances in the line. Note also that, to facilitate direct comparison among the methods, t(s) does not contain the time required for running the ILS.}

A ``-'' indicates that the value in the entry is not available because the model was not run on that group of instances. For AF and EAF we also report the average gap per million, computed as gap$_{\text{pm}}$=$10^6(U-L)/U$, with $U$ and $L$ being, respectively, the best upper and lower bound value obtained in the run (a ``-'' is reported \rev{when no valid $L$ is obtained, i.e., when even the LP relaxation of the model was not solved due to time or memory limits)}.

\begin{table}[htbp]
	\centering
	\caption{Results for \rev{set 1 instances with $p_{\max}=20$ (time limit = 300 seconds, ILS time not included)}}
	\scriptsize
	\setlength{\tabcolsep}{0.4mm}
	\rev{
		\begin{tabular}{lrrrrrrrrrrrrrrrrrrrrrrrrr}
			\toprule
			\multirow{3}[6]{*}{$n$} & \multirow{3}[6]{*}{$m$} &       & \multicolumn{13}{c}{Existing formulations}                                                            &       & \multicolumn{9}{c}{New formulations} \\
			\cmidrule{4-16}\cmidrule{18-26}          &       &       & \multicolumn{2}{c}{CIQP} &       & \multicolumn{2}{c}{PTI} &       & \multicolumn{3}{c}{SC} &       & \multicolumn{3}{c}{TI} &       & \multicolumn{4}{c}{AF}        &       & \multicolumn{4}{c}{EAF} \\
			\cmidrule{4-5}\cmidrule{7-8}\cmidrule{10-12}\cmidrule{14-16}\cmidrule{18-21}\cmidrule{23-26}          &       &       & {\#opt} & {t(s)} &       & {\#opt} & {t(s)} &       & {\#root} & {\#opt} & {t(s)} &       & {\#root} & {\#opt} & {t(s)} &       & {\#root} & {\#opt} & {t(s)} & {gap$_\text{pm}$} &       & {\#root} & {\#opt} & {t(s)} & {gap$_\text{pm}$} \\
			\midrule
			\multirow{4}[2]{*}{30} & 2     &       & 10    & 1.1   &       & 10    & 11.1  &       & 10    & 10    & 0.2   &       & 10    & 10    & 1.2   &       & 10    & 10    & 0.1   & 0.0   &       & 10    & 10    & 0.0   & 0.0 \\
			& 4     &       & 1     & 278.1 &       & 8     & 165.9 &       & 10    & 10    & 0.1   &       & 10    & 10    & 0.5   &       & 10    & 10    & 0.1   & 0.0   &       & 10    & 10    & 0.0   & 0.0 \\
			& 6     &       & 0     & t.lim &       & 3     & 240.1 &       & 10    & 10    & 0.0   &       & 10    & 10    & 0.2   &       & 10    & 10    & 0.0   & 0.0   &       & 10    & 10    & 0.0   & 0.0 \\
			& 8     &       & 0     & t.lim &       & 3     & 243.9 &       & 10    & 10    & 0.0   &       & 10    & 10    & 0.2   &       & 10    & 10    & 0.0   & 0.0   &       & 10    & 10    & 0.0   & 0.0 \\
			\midrule
			\multirow{6}[2]{*}{100} & 2     &       & 0     & t.lim &       & 0     & t.lim &       & 9     & 0     & t.lim &       & 10    & 10    & 199.8 &       & 10    & 10    & 1.6   & 0.0   &       & 10    & 10    & 0.9   & 0.0 \\
			& 4     &       & 0     & t.lim &       & 0     & t.lim &       & 10    & 3     & 276.5 &       & 10    & 10    & 65.3  &       & 10    & 10    & 2.1   & 0.0   &       & 10    & 10    & 0.4   & 0.0 \\
			& 6     &       & 0     & t.lim &       & 0     & t.lim &       & 10    & 10    & 101.3 &       & 10    & 10    & 29.5  &       & 10    & 10    & 0.5   & 0.0   &       & 10    & 10    & 0.4   & 0.0 \\
			& 8     &       & 0     & t.lim &       & 0     & t.lim &       & 10    & 10    & 15.9  &       & 10    & 10    & 19.3  &       & 10    & 10    & 0.4   & 0.0   &       & 10    & 10    & 0.3   & 0.0 \\
			& 16    &       & 0     & t.lim &       & 0     & t.lim &       & 10    & 10    & 2.2   &       & 10    & 10    & 3.9   &       & 10    & 10    & 0.1   & 0.0   &       & 10    & 10    & 0.1   & 0.0 \\
			& 30    &       & 0     & t.lim &       & 0     & t.lim &       & 10    & 10    & 0.4   &       & 10    & 10    & 1.2   &       & 10    & 10    & 0.1   & 0.0   &       & 10    & 10    & 0.1   & 0.0 \\
			\midrule
			\multirow{6}[2]{*}{400} & 2     &       & 0     & t.lim &       & 0     & t.lim &       & 0     & 0     & t.lim &       & 0     & 0     & m.lim &       & 10    & 10    & 15.0  & 0.0   &       & 10    & 10    & 4.8   & 0.0 \\
			& 4     &       & 0     & t.lim &       & 0     & t.lim &       & 0     & 0     & t.lim &       & 0     & 0     & m.lim &       & 10    & 10    & 29.3  & 0.0   &       & 10    & 10    & 17.0  & 0.0 \\
			& 6     &       & 0     & t.lim &       & 0     & t.lim &       & 0     & 0     & t.lim &       & 0     & 0     & m.lim &       & 10    & 10    & 32.3  & 0.0   &       & 10    & 10    & 13.2  & 0.0 \\
			& 8     &       & 0     & t.lim &       & 0     & t.lim &       & 0     & 0     & t.lim &       & 0     & 0     & m.lim &       & 10    & 10    & 33.7  & 0.0   &       & 10    & 10    & 10.6  & 0.0 \\
			& 16    &       & 0     & t.lim &       & 0     & t.lim &       & 0     & 0     & t.lim &       & 0     & 0     & m.lim &       & 10    & 10    & 16.2  & 0.0   &       & 10    & 10    & 3.5   & 0.0 \\
			& 30    &       & 0     & t.lim &       & 0     & t.lim &       & 10    & 0     & t.lim &       & 10    & 3     & 272.9 &       & 10    & 10    & 3.9   & 0.0   &       & 10    & 10    & 0.9   & 0.0 \\
			\midrule
			\multirow{6}[2]{*}{700} & 2     &       & -     & -     &       & -     & -     &       & -     & -     & -     &       & -     & -     & -     &       & 10    & 6     & 191.3 & 0.1   &       & 10    & 10    & 66.8  & 0.0 \\
			& 4     &       & -     & -     &       & -     & -     &       & -     & -     & -     &       & -     & -     & -     &       & 10    & 7     & 220.2 & 0.4   &       & 10    & 10    & 56.9  & 0.0 \\
			& 6     &       & -     & -     &       & -     & -     &       & -     & -     & -     &       & -     & -     & -     &       & 10    & 4     & 240.6 & 1.7   &       & 10    & 10    & 55.4  & 0.0 \\
			& 8     &       & -     & -     &       & -     & -     &       & -     & -     & -     &       & -     & -     & -     &       & 10    & 6     & 185.0 & 1.9   &       & 10    & 10    & 86.4  & 0.0 \\
			& 16    &       & -     & -     &       & -     & -     &       & -     & -     & -     &       & -     & -     & -     &       & 10    & 9     & 202.3 & 0.2   &       & 10    & 10    & 20.3  & 0.0 \\
			& 30    &       & -     & -     &       & -     & -     &       & -     & -     & -     &       & -     & -     & -     &       & 10    & 10    & 39.9  & 0.0   &       & 10    & 10    & 3.1   & 0.0 \\
			\midrule
			\multirow{6}[2]{*}{1000} & 2     &       & -     & -     &       & -     & -     &       & -     & -     & -     &       & -     & -     & -     &       & 10    & 0     & t.lim & 0.3   &       & 10    & 10    & 105.9 & 0.0 \\
			& 4     &       & -     & -     &       & -     & -     &       & -     & -     & -     &       & -     & -     & -     &       & 10    & 0     & t.lim & 1.5   &       & 10    & 10    & 89.6  & 0.0 \\
			& 6     &       & -     & -     &       & -     & -     &       & -     & -     & -     &       & -     & -     & -     &       & 10    & 3     & 287.2 & 2.1   &       & 10    & 10    & 139.6 & 0.0 \\
			& 8     &       & -     & -     &       & -     & -     &       & -     & -     & -     &       & -     & -     & -     &       & 10    & 3     & 278.7 & 1.8   &       & 10    & 10    & 91.6  & 0.0 \\
			& 16    &       & -     & -     &       & -     & -     &       & -     & -     & -     &       & -     & -     & -     &       & 10    & 0     & 300.0 & 8.6   &       & 10    & 10    & 64.2  & 0.0 \\
			& 30    &       & -     & -     &       & -     & -     &       & -     & -     & -     &       & -     & -     & -     &       & 10    & 4     & 264.3 & 8.2   &       & 10    & 10    & 24.0  & 0.0 \\
			\midrule
			total/avg &       &       & 11    & 288.5 &       & 24    & 266.3 &       & 109   & 83    & 156.0 &       & 110   & 103   & 54.0  &       & 280   & 212   & 105.2 & 1.0   &       & 280   & 280   & 30.6  & 0.0 \\
			\bottomrule
		\end{tabular}%
	}
	\label{tab:res_p20}
\end{table}%

\begin{table}[htbp]
	\centering
	\caption{Results for \rev{set 1 instances with $p_{\max}=100$ (time limit = 300 seconds, ILS time not included)}}
	\scriptsize
	\setlength{\tabcolsep}{0.37mm}
	\rev{
		\begin{tabular}{lrrrrrrrrrrrrrrrrrrrrrrrrr}
			\toprule
			\multirow{3}[6]{*}{$n$} & \multirow{3}[6]{*}{$m$} &       & \multicolumn{13}{c}{Existing formulations}                                                            &       & \multicolumn{9}{c}{New formulations} \\
			\cmidrule{4-16}\cmidrule{18-26}          &       &       & \multicolumn{2}{c}{CIQP} &       & \multicolumn{2}{c}{PTI} &       & \multicolumn{3}{c}{SC} &       & \multicolumn{3}{c}{TI} &       & \multicolumn{4}{c}{AF}        &       & \multicolumn{4}{c}{EAF} \\
			\cmidrule{4-5}\cmidrule{7-8}\cmidrule{10-12}\cmidrule{14-16}\cmidrule{18-21}\cmidrule{23-26}          &       &       & {\#opt} & {t(s)} &       & {\#opt} & {t(s)} &       & {\#root} & {\#opt} & {t(s)} &       & {\#root} & {\#opt} & {t(s)} &       & {\#root} & {\#opt} & {t(s)} & {gap$_\text{pm}$} &       & {\#root} & {\#opt} & {t(s)} & {gap$_\text{pm}$} \\
			\midrule
			\multirow{4}[2]{*}{30} & 2     &       & 10    & 0.8   &       & 10    & 115.0 &       & 10    & 10    & 0.2   &       & 10    & 10    & 73.3  &       & 10    & 10    & 0.5   & 0.0   &       & 10    & 10    & 0.3   & 0.0 \\
			& 4     &       & 0     & t.lim &       & 0     & t.lim &       & 10    & 10    & 0.1   &       & 10    & 10    & 18.6  &       & 10    & 10    & 0.3   & 0.0   &       & 10    & 10    & 0.2   & 0.0 \\
			& 6     &       & 0     & t.lim &       & 0     & t.lim &       & 10    & 10    & 0.0   &       & 10    & 10    & 6.9   &       & 10    & 10    & 0.2   & 0.0   &       & 10    & 10    & 0.2   & 0.0 \\
			& 8     &       & 0     & t.lim &       & 0     & t.lim &       & 10    & 10    & 0.0   &       & 10    & 10    & 4.4   &       & 10    & 10    & 0.2   & 0.0   &       & 10    & 10    & 0.1   & 0.0 \\
			\midrule
			\multirow{6}[2]{*}{100} & 2     &       & 0     & t.lim &       & 0     & t.lim &       & 0     & 0     & t.lim &       & 0     & 0     & m.lim &       & 10    & 10    & 10.7  & 0.0   &       & 10    & 10    & 7.6   & 0.0 \\
			& 4     &       & 0     & t.lim &       & 0     & t.lim &       & 10    & 3     & 257.2 &       & 0     & 0     & m.lim &       & 10    & 10    & 36.5  & 0.2   &       & 10    & 10    & 25.9  & 0.2 \\
			& 6     &       & 0     & t.lim &       & 0     & t.lim &       & 10    & 4     & 224.5 &       & 0     & 0     & m.lim &       & 10    & 10    & 14.5  & 0.0   &       & 10    & 10    & 8.3   & 0.0 \\
			& 8     &       & 0     & t.lim &       & 0     & t.lim &       & 10    & 3     & 239.7 &       & 0     & 0     & t.lim &       & 10    & 10    & 8.1   & 0.0   &       & 10    & 10    & 10.3  & 0.0 \\
			& 16    &       & 0     & t.lim &       & 0     & t.lim &       & 10    & 8     & 86.9  &       & 9     & 6     & 264.0 &       & 10    & 10    & 2.5   & 0.0   &       & 10    & 10    & 1.5   & 0.0 \\
			& 30    &       & 0     & t.lim &       & 0     & t.lim &       & 10    & 10    & 1.7   &       & 10    & 10    & 49.3  &       & 10    & 10    & 0.6   & 0.0   &       & 10    & 10    & 0.5   & 0.0 \\
			\midrule
			\multirow{6}[2]{*}{400} & 2     &       & 0     & t.lim &       & 0     & t.lim &       & 0     & 0     & t.lim &       & 0     & 0     & m.lim &       & 10    & 0     & t.lim & 0.6   &       & 10    & 10    & 264.6 & 0.0 \\
			& 4     &       & 0     & t.lim &       & 0     & t.lim &       & 0     & 0     & t.lim &       & 0     & 0     & m.lim &       & 10    & 2     & 292.5 & 2.6   &       & 10    & 10    & 180.7 & 0.0 \\
			& 6     &       & 0     & t.lim &       & 0     & t.lim &       & 0     & 0     & t.lim &       & 0     & 0     & m.lim &       & 10    & 7     & 272.2 & 2.6   &       & 10    & 10    & 176.6 & 0.0 \\
			& 8     &       & 0     & t.lim &       & 0     & t.lim &       & 0     & 0     & t.lim &       & 0     & 0     & m.lim &       & 10    & 9     & 233.3 & 1.0   &       & 10    & 10    & 177.5 & 0.0 \\
			& 16    &       & 0     & t.lim &       & 0     & t.lim &       & 0     & 0     & t.lim &       & 0     & 0     & m.lim &       & 10    & 8     & 181.4 & 4.6   &       & 10    & 10    & 111.8 & 0.0 \\
			& 30    &       & 0     & t.lim &       & 0     & t.lim &       & 0     & 0     & t.lim &       & 0     & 0     & m.lim &       & 10    & 10    & 94.6  & 0.0   &       & 10    & 10    & 64.1  & 0.0 \\
			\midrule
			\multirow{6}[2]{*}{700} & 2     &       & -     & -     &       & -     & -     &       & -     & -     & -     &       & -     & -     & -     &       & 0     & 0     & t.lim & -     &       & 9     & 0     & t.lim & 0.8 \\
			& 4     &       & -     & -     &       & -     & -     &       & -     & -     & -     &       & -     & -     & -     &       & 6     & 0     & t.lim & 3.2   &       & 9     & 0     & t.lim & 3.3 \\
			& 6     &       & -     & -     &       & -     & -     &       & -     & -     & -     &       & -     & -     & -     &       & 6     & 0     & t.lim & 5.5   &       & 10    & 0     & t.lim & 5.5 \\
			& 8     &       & -     & -     &       & -     & -     &       & -     & -     & -     &       & -     & -     & -     &       & 10    & 0     & t.lim & 8.4   &       & 10    & 0     & t.lim & 8.4 \\
			& 16    &       & -     & -     &       & -     & -     &       & -     & -     & -     &       & -     & -     & -     &       & 10    & 0     & t.lim & 16.0  &       & 10    & 0     & t.lim & 16.0 \\
			& 30    &       & -     & -     &       & -     & -     &       & -     & -     & -     &       & -     & -     & -     &       & 10    & 0     & t.lim & 28.5  &       & 10    & 0     & t.lim & 28.4 \\
			\midrule
			\multirow{6}[2]{*}{1000} & 2     &       & -     & -     &       & -     & -     &       & -     & -     & -     &       & -     & -     & -     &       & 0     & 0     & m.lim & -     &       & 0     & 0     & m.lim & - \\
			& 4     &       & -     & -     &       & -     & -     &       & -     & -     & -     &       & -     & -     & -     &       & 0     & 0     & m.lim & -     &       & 0     & 0     & m.lim & - \\
			& 6     &       & -     & -     &       & -     & -     &       & -     & -     & -     &       & -     & -     & -     &       & 0     & 0     & m.lim & -     &       & 0     & 0     & m.lim & - \\
			& 8     &       & -     & -     &       & -     & -     &       & -     & -     & -     &       & -     & -     & -     &       & 7     & 0     & t.lim & 6.4   &       & 9     & 0     & t.lim & 6.4 \\
			& 16    &       & -     & -     &       & -     & -     &       & -     & -     & -     &       & -     & -     & -     &       & 10    & 0     & t.lim & 13.5  &       & 10    & 0     & t.lim & 13.5 \\
			& 30    &       & -     & -     &       & -     & -     &       & -     & -     & -     &       & -     & -     & -     &       & 10    & 0     & t.lim & 23.0  &       & 10    & 0     & t.lim & 23.0 \\
			\midrule
			total/avg &       &       & 10    & 289.3 &       & 10    & 288.4 &       & 90    & 68    & 181.9 &       & 53    & 53    & 102.4 &       & 229   & 136   & 165.9 & 4.8   &       & 247   & 160   & 149.2 & 4.2 \\
			\bottomrule
		\end{tabular}%
	}
	\label{tab:res_p100}
\end{table}%

The results show that EAF clearly outperforms all other methods on the attempted instances. It solves to proven optimality all instances with $p_{\max} = 20$ and all instances with $p_{\max} = 100$ and $n \leq 400$, so $440$ out of the $560$ tested instances. The version without enhancements, AF, solves all the instances with up to $400$ jobs for $p_{\max} = 20$ and  up to $100$ jobs for $p_{\max} = 100$, for a total of $348$ out of $560$ instances. This proves that the enhancements presented in Section \ref{sec:method} are very effective. For the unsolved instances, the gaps are extremely small, \rev{amounting to just a few units} per million on average.

The B\&P implemented to solve SC fails in solving some instances with just 100 jobs, which is coherent with the results in \cite{vandenAkkeretal1999}. The same happens for TI, which fails in solving instances with 100 jobs when $p_{\max} = 100$ and $m$ is small. The performance of both CIPQ and PTI is  very poor, but that could be explained by the fact that the two models were originally developed for the $R||\sum w_jC_j$, and hence do not exploit the symmetries induced by the identical machines.
The main advantage of CIPQ is related to its polynomial size, but it optimally solves only few instances with just $30$ jobs.
For what concerns PTI, it is worth mentioning that it was solved by means of a Benders decomposition method
in \cite{BulbulandSen2017}, as discussed in Section \ref{sec:PeemptionTI}. The authors provided us with the results obtained by their method on our $P||\sum w_jC_j$ instances. Unfortunately, these were worse on average that those obtained by the direct solution of the PTI model by means of the Gurobi solver. This, once more, can be imputed to the fact that their method was developed for the case of unrelated machines.

\rev{We conducted further experiments by running the models under larger time limits up to one hour. The summary of the results that we obtained is presented in Table \ref{tab:new_opts}, which shows, for each model and for each time limit, the number of optimal solutions found under the different time limits. Once more, it can be observed that AF and EAF clearly outperform the other methods. It can also be seen that the enhancements discussed in Section \ref{sec:method} are indeed effective, because EAF solved to proven optimality all instances with $p_{\max} = 20$ in less than $300$ seconds, whereas AF did not solve two of them within $1$ hour. Concerning the instances with $p_{\max} = 100$, by increasing the time limit from $300$ to $3600$ seconds, AF was able to solve $61$ instances more, including all the unsolved instances with  $n = 400$, whereas EAF solved $62$ more, including $48$ out of $60$ instances with $n=700$ that where not solved within $300$ seconds. In total, EAF was able to optimally solve $502$ out of $560$ instances.
}
\begin{table}
	\caption{\rev{Number of optimal solutions found under different time limits}}
	\centering
	\subfloat[\rev{$p_{\max} = 20$ (280 instances)}]{
		\centering
		\scriptsize
		\setlength{\tabcolsep}{0.75mm}
		\rev{
			\begin{tabular}{rrrrrrrrrrrrr}
				\toprule
				\multirow{2}{0.7cm}{time limit}  &       & \multicolumn{11}{c}{Formulation} \\
				\cmidrule{3-13}
				&       & CQIP  &       & PTI   &       & SC    &       & TI    &       & AF    &       & EAF \\
				\cmidrule{1-1}\cmidrule{3-3}\cmidrule{5-5}\cmidrule{7-7}\cmidrule{9-9}\cmidrule{11-11}\cmidrule{13-13}
				300   &       & 11    &       & 24    &       & 83    &       & 103   &       & 212   &       & 280 \\
				600   &       & 12    &       & 28    &       & 86    &       & 110   &       & 247   &       & 280 \\
				900   &       & 12    &       & 31    &       & 89    &       & 110   &       & 254   &       & 280 \\
				1200  &       & 12    &       & 32    &       & 90    &       & 110   &       & 260   &       & 280 \\
				1500  &       & 13    &       & 32    &       & 90    &       & 110   &       & 265   &       & 280 \\
				1800  &       & 14    &       & 32    &       & 90    &       & 110   &       & 268   &       & 280 \\
				2100  &       & 14    &       & 32    &       & 90    &       & 110   &       & 272   &       & 280 \\
				2400  &       & 15    &       & 32    &       & 92    &       & 110   &       & 275   &       & 280 \\
				2700  &       & 15    &       & 32    &       & 95    &       & 110   &       & 276   &       & 280 \\
				3000  &       & 15    &       & 32    &       & 98    &       & 110   &       & 276   &       & 280 \\
				3300  &       & 15    &       & 32    &       & 98    &       & 110   &       & 277   &       & 280 \\
				3600  &       & 15    &       & 32    &       & 99    &       & 110   &       & 278   &       & 280 \\
				\bottomrule
			\end{tabular}%
		}
		\label{tab:new_opts_pmax_20}%
	}\quad
	\subfloat[\rev{$p_{\max} = 100$ (280 instances)}]{
		\centering
		\scriptsize
		\setlength{\tabcolsep}{0.75mm}
		\rev{
			\begin{tabular}{rrrrrrrrrrrrr}
				\toprule
				\multirow{2}{0.7cm}{time limit}  &       & \multicolumn{11}{c}{Formulation} \\
				\cmidrule{3-13}
				&       & CQIP  &       & PTI   &       & SC    &       & TI    &       & AF    &       & EAF \\
				\cmidrule{1-1}\cmidrule{3-3}\cmidrule{5-5}\cmidrule{7-7}\cmidrule{9-9}\cmidrule{11-11}\cmidrule{13-13}
				300   &       & 10    &       & 10    &       & 68    &       & 53    &       & 136   &       & 160 \\
				600   &       & 11    &       & 10    &       & 70    &       & 60    &       & 150   &       & 169 \\
				900   &       & 12    &       & 10    &       & 73    &       & 61    &       & 168   &       & 178 \\
				1200  &       & 13    &       & 10    &       & 73    &       & 61    &       & 173   &       & 180 \\
				1500  &       & 13    &       & 10    &       & 74    &       & 62    &       & 179   &       & 192 \\
				1800  &       & 14    &       & 10    &       & 76    &       & 64    &       & 180   &       & 200 \\
				2100  &       & 15    &       & 10    &       & 77    &       & 65    &       & 183   &       & 206 \\
				2400  &       & 15    &       & 10    &       & 77    &       & 66    &       & 184   &       & 210 \\
				2700  &       & 15    &       & 10    &       & 78    &       & 68    &       & 188   &       & 213 \\
				3000  &       & 15    &       & 10    &       & 78    &       & 68    &       & 192   &       & 216 \\
				3300  &       & 15    &       & 10    &       & 79    &       & 68    &       & 194   &       & 219 \\
				3600  &       & 15    &       & 10    &       & 80    &       & 69    &       & 197   &       & 222 \\
				\bottomrule
			\end{tabular}%
		}
		\label{tab:new_opts_pmax_100}%
	}
	\label{tab:new_opts}%
\end{table}			

Time indexed formulations, such as PTI, TI, AF and EAF, have a pseudo-polynomial size, and hence may require large amounts of memory when the time horizon grows. It can be noticed indeed that all of them do not solve some instances due to memory limit. In this sense, TI starts to run out of memory for instances with only $100$ jobs, PTI is able to deal with instances with up to $400$ jobs (although cannot optimally solve them due to time limit), whereas AF and EAF deal with instances with up to $1000$ jobs.

This better behavior can be justified by the difference in the number of variables required by models. This fact is graphically highlighted in Figure \ref{fig:var_reduction}, which presents the average number of variables in thousands, per group of instances having the same $n$. The reduction of AF and even more EAF with respect to the plain TI model is evident.
{This effect can be observed more in details in Table \ref{tab:res_var}, which reports, for each group of instances having same $n$ and $m$, the number of variables in thousands, var(thousands), and the percentage reduction of variables from one model to the next, {red(\%)}.}
\begin{figure}[!h]
	\centering
	\subfloat[$p_{\max}=20$]{\includegraphics[width=0.49\textwidth]{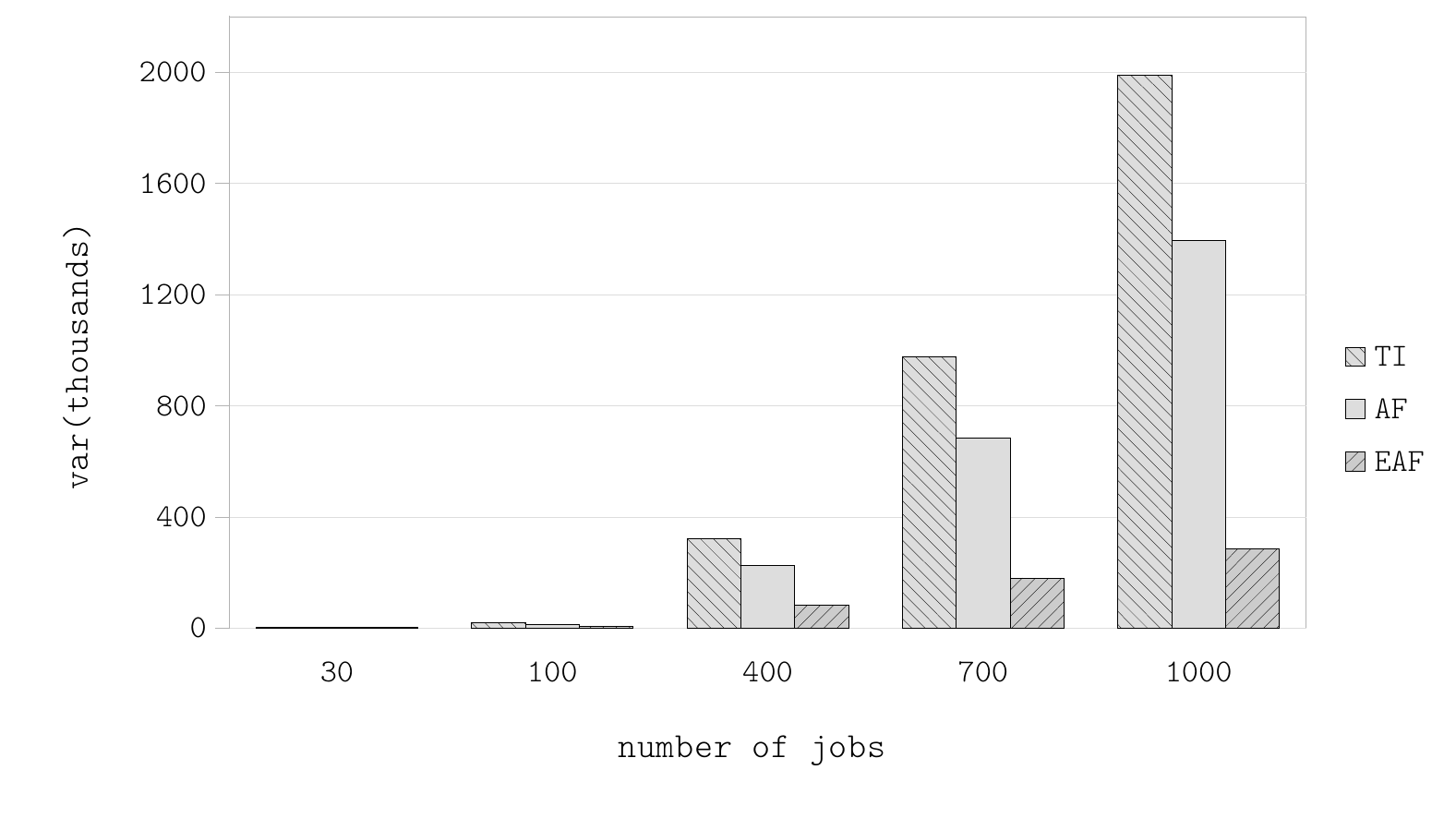}\label{fig:var_pmax20}}
	\hfill
	\subfloat[$p_{\max}=100$]{\includegraphics[width=0.49\textwidth]{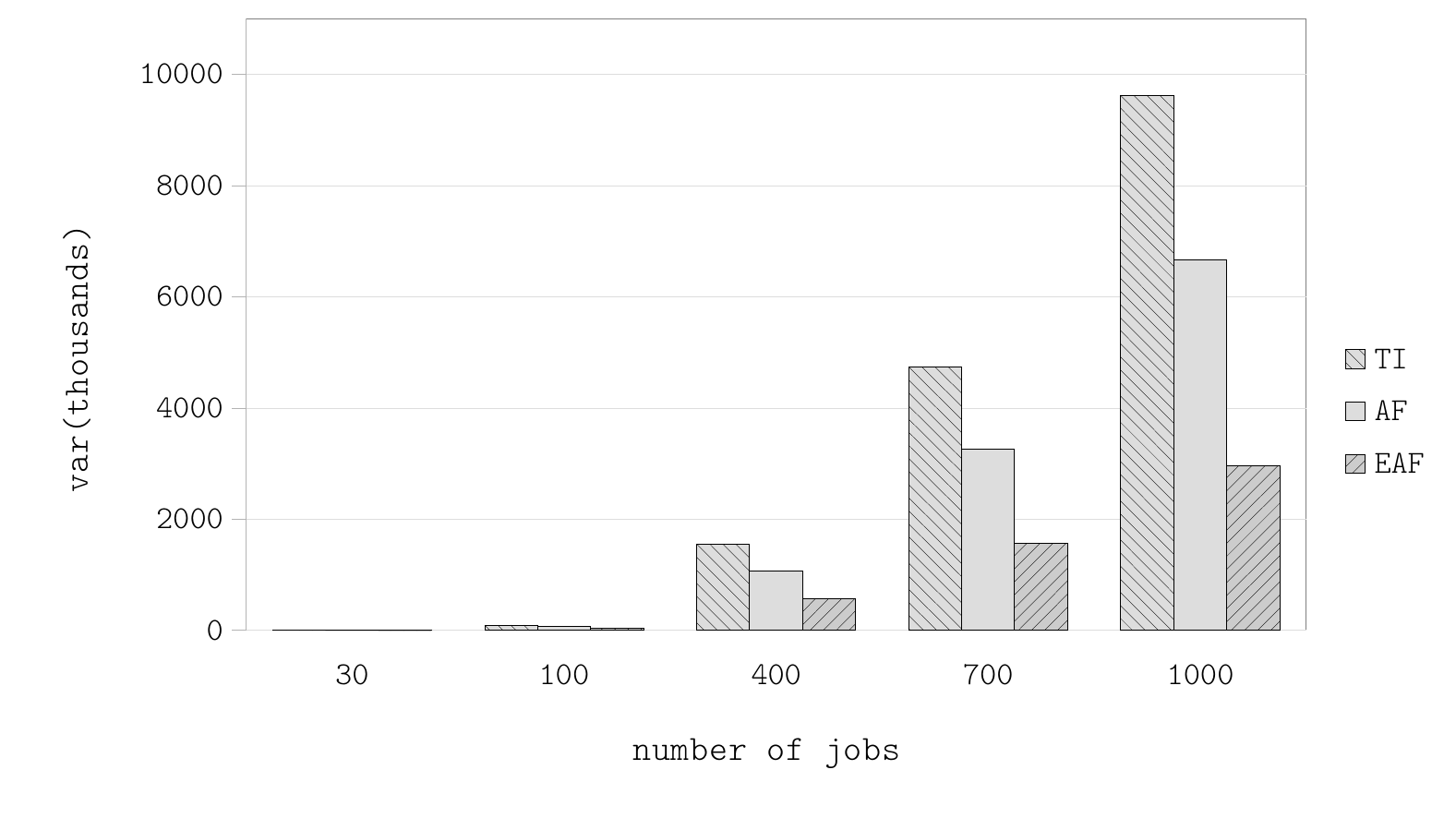}\label{fig:var_pmax100}}
	\caption{Impact of the proposed enhancements on the number of variables}
	\label{fig:var_reduction}
\end{figure}			
From Table \ref{tab:res_var}, it is possible to notice that for the instances with $p_{\max}=20$ reductions of about 30\% are obtained by AF over TI, and of even 80\% by EAF over AF. This allows to move from an average of about $7 \times 10^5$ TI variables to just $1 \times 10^5$ EAF variables. Concerning the instances with $p_{\max}=100$, EAF formulates the problem using, on average, $50\%$ less variables than AF, which in turns uses 30\% less variables than TI. The results  also indicates that, as the instance grows, the reductions become \rev{larger}.

\begin{table}[!h]
	\centering
	\caption{Variables (in thousands) required by the main pseudo-polynomial formulations}
	\scriptsize
	\setlength{\tabcolsep}{0.5mm}
	\begin{tabular}{lrrrrrrrrrrrrrrrrrrrr}
		\toprule
		\multicolumn{1}{c}{\multirow{3}[5]{*}{$n$}} & \multicolumn{1}{c}{\multirow{3}[5]{*}{$m$}} & \multicolumn{9}{c}{$p_{\max}=20$} &  & \multicolumn{9}{c}{$p_{\max}=100$} \\
		\cmidrule(){3-11} \cmidrule(){13-21}
		&  & \multicolumn{5}{c}{{var(thousands)}} &  & \multicolumn{3}{c}{red(\%)} &  & \multicolumn{5}{c}{{var(thousands)}} &  & \multicolumn{3}{c}{red(\%)} \\
		\cmidrule(){3-7} \cmidrule(){9-11} \cmidrule(){13-17} \cmidrule(){19-21}
		&  & \multicolumn{1}{r}{TI} &  & \multicolumn{1}{r}{AF} &  & \multicolumn{1}{r}{EAF} &  & \multicolumn{1}{r}{AF {\emph vs} TI} &  & \multicolumn{1}{r}{EAF {\emph vs} AF} &  & \multicolumn{1}{r}{TI} &  & \multicolumn{1}{r}{AF} &  & \multicolumn{1}{r}{EAF} &  & \multicolumn{1}{r}{AF {\emph vs} TI} &  & \multicolumn{1}{r}{EAF {\emph vs} AF} \\
		\cline{3-3} \cline{5-5} \cline{7-7} \cline{9-9} \cline{11-11} \cline{13-13} \cline{15-15} \cline{17-17} \cline{19-19} \cline{21-21}
		\hline
		\multirow{4}{*}{30} & 2 & 4.8 &  & 3.0 &  & 1.8 &  & 37.5 &  & 41.0 &  & 20.8 &  & 11.8 &  & 7.6 &  & 43.2 &  & 35.5 \\
		& 4 & 2.6 &  & 1.8 &  & 1.3 &  & 29.2 &  & 29.4 &  & 11.9 &  & 7.8 &  & 5.7 &  & 34.4 &  & 27.4 \\
		& 6 & 1.8 &  & 1.3 &  & 1.0 &  & 24.8 &  & 28.6 &  & 8.4 &  & 5.8 &  & 4.5 &  & 31.3 &  & 21.2 \\
		& 8 & 1.4 &  & 1.1 &  & 0.8 &  & 22.2 &  & 24.8 &  & 6.7 &  & 4.6 &  & 3.7 &  & 32.4 &  & 19.6 \\
		\hline
		\multirow{6}{*}{100} & 2 & 51.9 &  & 31.4 &  & 17.0 &  & 39.5 &  & 46.0 &  & 253.9 &  & 150.7 &  & 89.7 &  & 40.7 &  & 40.5 \\
		& 4 & 27.6 &  & 20.1 &  & 10.8 &  & 27.0 &  & 46.5 &  & 128.0 &  & 90.5 &  & 58.3 &  & 29.2 &  & 35.6 \\
		& 6 & 17.7 &  & 13.6 &  & 7.6 &  & 23.5 &  & 44.2 &  & 90.1 &  & 69.8 &  & 44.2 &  & 22.5 &  & 36.6 \\
		& 8 & 13.7 &  & 11.1 &  & 6.7 &  & 19.3 &  & 39.9 &  & 66.4 &  & 51.9 &  & 33.7 &  & 21.9 &  & 35.0 \\
		& 16 & 7.5 &  & 6.4 &  & 4.1 &  & 13.9 &  & 36.3 &  & 36.1 &  & 29.7 &  & 21.8 &  & 17.8 &  & 26.5 \\
		& 30 & 4.4 &  & 3.8 &  & 2.7 &  & 12.1 &  & 30.4 &  & 21.7 &  & 18.7 &  & 14.1 &  & 13.8 &  & 24.4 \\
		\hline
		\multirow{6}{*}{400} & 2 & 843.2 &  & 509.4 &  & 185.7 &  & 39.6 &  & 63.5 &  & 4088.9 &  & 2409.8 &  & 1291.0 &  & 41.1 &  & 46.4 \\
		& 4 & 424.3 &  & 309.1 &  & 108.0 &  & 27.2 &  & 65.1 &  & 2017.3 &  & 1418.7 &  & 726.3 &  & 29.7 &  & 48.8 \\
		& 6 & 282.6 &  & 220.6 &  & 75.8 &  & 21.9 &  & 65.7 &  & 1360.9 &  & 1034.0 &  & 530.2 &  & 24.0 &  & 48.7 \\
		& 8 & 213.2 &  & 172.6 &  & 60.8 &  & 19.0 &  & 64.8 &  & 1040.4 &  & 827.6 &  & 429.4 &  & 20.5 &  & 48.1 \\
		& 16 & 106.9 &  & 92.6 &  & 34.4 &  & 13.4 &  & 62.9 &  & 533.0 &  & 456.4 &  & 247.8 &  & 14.4 &  & 45.7 \\
		& 30 & 60.2 &  & 54.6 &  & 22.4 &  & 9.4 &  & 59.0 &  & 290.0 &  & 256.2 &  & 151.4 &  & 11.7 &  & 40.9 \\
		\hline
		\multirow{6}{*}{700} & 2 & 2539.6 &  & 1520.0 &  & 418.7 &  & 40.1 &  & 72.5 &  & 12613.7 &  & 7437.6 &  & 3720.1 &  & 41.0 &  & 50.0 \\
		& 4 & 1282.9 &  & 930.0 &  & 240.6 &  & 27.5 &  & 74.1 &  & 6114.2 &  & 4336.5 &  & 2073.8 &  & 29.1 &  & 52.2 \\
		& 6 & 874.3 &  & 680.5 &  & 170.7 &  & 22.2 &  & 74.9 &  & 4161.0 &  & 3191.1 &  & 1464.8 &  & 23.3 &  & 54.1 \\
		& 8 & 654.0 &  & 528.4 &  & 132.2 &  & 19.2 &  & 75.0 &  & 3148.3 &  & 2514.4 &  & 1140.7 &  & 20.1 &  & 54.6 \\
		& 16 & 326.9 &  & 284.2 &  & 75.2 &  & 13.1 &  & 73.5 &  & 1582.9 &  & 1355.0 &  & 639.9 &  & 14.4 &  & 52.8 \\
		& 30 & 177.4 &  & 160.9 &  & 45.7 &  & 9.3 &  & 71.6 &  & 859.6 &  & 766.9 &  & 386.1 &  & 10.8 &  & 49.6 \\
		\hline
		\multirow{6}{*}{1000} & 2 & 5216.2 &  & 3129.4 &  & 671.1 &  & 40.0 &  & 78.6 &  & 25320.3 &  & 14942.6 &  & 7020.6 &  & 41.0 &  & 53.0 \\
		& 4 & 2626.0 &  & 1902.1 &  & 378.1 &  & 27.6 &  & 80.1 &  & 12725.0 &  & 9083.5 &  & 3980.6 &  & 28.6 &  & 56.2 \\
		& 6 & 1757.3 &  & 1369.4 &  & 267.0 &  & 22.1 &  & 80.5 &  & 8445.0 &  & 6673.8 &  & 2853.3 &  & 21.0 &  & 57.2 \\
		& 8 & 1314.7 &  & 1063.1 &  & 207.7 &  & 19.1 &  & 80.5 &  & 6327.7 &  & 5051.4 &  & 2090.2 &  & 20.2 &  & 58.6 \\
		& 16 & 664.4 &  & 579.4 &  & 115.4 &  & 12.8 &  & 80.1 &  & 3212.8 &  & 2742.8 &  & 1154.8 &  & 14.6 &  & 57.9 \\
		& 30 & 359.4 &  & 326.5 &  & 69.6 &  & 9.2 &  & 78.7 &  & 1711.3 &  & 1531.0 &  & 688.2 &  & 10.5 &  & 55.1 \\
		\hline
		\multicolumn{2}{c}{total/avg} & 709.2 &  & 497.4 &  & 119.0 &  & 29.9 &  & 76.1 &  & 3435.6 &  & 2365.3 &  & 1098.6 &  & 31.2 &  & 53.6 \\
		\hline
	\end{tabular}
	\label{tab:res_var}
\end{table}		

\subsection{\rev{Computational results on benchmark set 2}}\label{sec:newResults}

\rev{Very recently, \cite{KowalczykandLeus2018} improved the branch-and-price method of \cite{vandenAkkeretal1999} by introducing the use of stabilization techniques, generic branching, and a {\em zero-suppressed binary decision diagram} (ZDDs) for solving the pricing subproblem. By combining these techniques, they devised three main methods: the first, named VHV-DP, uses the branching scheme and the DP in \cite{vandenAkkeretal1999}, but includes stabilization; the second, VHV-ZDD, also uses the branching scheme of \cite{vandenAkkeretal1999} and stabilization, but solves the pricing subproblem with the the ZDDs technique; the third, RF-ZDD, differs from the second by the fact that the branching decisions follow the generic scheme of \cite{RyanandFoster1981}. The three methods were computationally tested on the benchmark set 2 described in Section \ref{subsec:instances}.}	

\rev{We performed experiments by running for 600 seconds our best mathematical formulation, namely, EAF, on the same instances and compared our results with those obtained by \cite{KowalczykandLeus2018}. As in the previous section, for instances involving more than 100 jobs EAF used the ILS of Section \ref{subsec:UB} to obtain an initial solution.
According to the single thread results in \url{https://www.cpubenchmark.net/}, the processor used by \cite{KowalczykandLeus2018}, an Intel Core i7-3770 $3.40$ GHz,  is about $1.9$ times faster than our Intel Xeon E5530 $2.40$ GHz processor.}

\rev{The results that we obtained are presented in Tables \ref{tab:resnew_Class_I_II}, \ref{tab:resnew_Class_III_IV} and \ref{tab:resnew_Class_V_VI}. Columns {\#opt} and {t(s)} represent, for each method, class and group of $20$ instances, the number of instances solved to proven optimality and the average computational time, respectively.}	
\rev{Concerning RF-ZDD, VHV-ZDD and VHV-DP, the t(s) values report the average times in \cite{KowalczykandLeus2018} (i.e., not multiplied by 1.9), but now include in the computation of the average the entire time limit value ($3600$ seconds) for those instances that were not solved to proven optimality.}
\rev{The three methods by \cite{KowalczykandLeus2018} obtained good results but could not solve all instances to proven optimality. In general, VHV-DP was able to solve more instances than RF-ZDD and VHV-ZDD, whereas RF-ZDD proved to be less time consuming on average. We can observe from the tables that our ILS+EAF algorithm generally outperformed the other methods. Indeed, it could solve all $2400$ instances within the time limit and usually very quickly (only 10 instances required more than 300 seconds, and the slowest case required 395 seconds). Instances with up to 100 jobs were solved in a matter of seconds. Instances with 150 jobs required longer times, also due to the 100 seconds allowed for the ILS execution. The instances from class VI represent the most challenging testbed for ILS+EAF, nevertheless they were all solved in about 50 seconds on average.}

\begin{table}[htbp]
	\centering
	\caption{\rev{Results on set 2 instances -- Classes I and II}}
	\scriptsize
	\setlength{\tabcolsep}{0.5mm}
	\rev{
		\begin{tabular}{llrrrrrrrrrrrrrrrrrrrrrrrr}
			\toprule
			\multirow{3}[6]{*}{$n$} & \multicolumn{1}{r}{\multirow{3}[6]{*}{$m$}} &       & \multicolumn{11}{c}{Class I} &       & \multicolumn{11}{c}{Class II} \\
			\cmidrule{4-14}\cmidrule{16-26}
			&       &       & \multicolumn{2}{c}{RF-ZDD} &       & \multicolumn{2}{c}{VHV-ZDD} &       & \multicolumn{2}{c}{VHV-DP} &       & \multicolumn{2}{c}{ILS+EAF} &       & \multicolumn{2}{c}{RF-ZDD} &       & \multicolumn{2}{c}{VHV-ZDD} &       & \multicolumn{2}{c}{VHV-DP} &       & \multicolumn{2}{c}{ILS+EAF} \\
			\cmidrule{4-5}\cmidrule{7-8}\cmidrule{10-11}\cmidrule{13-14}\cmidrule{16-17}\cmidrule{19-20}\cmidrule{22-23}\cmidrule{25-26}
			&       &       & \#opt & t(s)  &       & \#opt & t(s)  &       & \#opt & t(s)  &       & \#opt & t(s)  &       & \#opt & t(s)   &       & \#opt & t(s)   &       & \#opt & t(s)   &       & \#opt & t(s) \\
			\midrule
			\multirow{5}[2]{*}{20} & \multicolumn{1}{r}{3} &       & 20    & 0.0   &       & 20    & 0.0   &       & 20    & 0.0   &       & 20    & 0.0   &       & 20    & 0.0   &       & 20    & 0.0   &       & 20    & 0.0   &       & 20    & 0.1 \\
			& \multicolumn{1}{r}{5} &       & 20    & 0.0   &       & 20    & 0.0   &       & 20    & 0.0   &       & 20    & 0.0   &       & 20    & 0.0   &       & 20    & 0.0   &       & 20    & 0.0   &       & 20    & 0.1 \\
			& \multicolumn{1}{r}{8} &       & 20    & 0.0   &       & 20    & 0.0   &       & 20    & 0.0   &       & 20    & 0.0   &       & 20    & 0.0   &       & 20    & 0.0   &       & 20    & 0.0   &       & 20    & 0.0 \\
			& \multicolumn{1}{r}{10} &       & 20    & 0.0   &       & 20    & 0.0   &       & 20    & 0.0   &       & 20    & 0.0   &       & 20    & 0.0   &       & 20    & 0.0   &       & 20    & 0.0   &       & 20    & 0.0 \\
			& \multicolumn{1}{r}{12} &       & 20    & 0.0   &       & 20    & 0.0   &       & 20    & 0.0   &       & 20    & 0.0   &       & 20    & 0.0   &       & 20    & 0.0   &       & 20    & 0.0   &       & 20    & 0.0 \\
			\midrule
			\multirow{5}[2]{*}{50} & \multicolumn{1}{r}{3} &       & 20    & 1.9   &       & 20    & 1.2   &       & 20    & 1.1   &       & 20    & 0.0   &       & 20    & 0.7   &       & 20    & 0.6   &       & 20    & 0.7   &       & 20    & 0.9 \\
			& \multicolumn{1}{r}{5} &       & 20    & 0.9   &       & 20    & 0.5   &       & 20    & 0.6   &       & 20    & 0.0   &       & 20    & 0.6   &       & 19    & 180.6 &       & 20    & 0.6   &       & 20    & 0.7 \\
			& \multicolumn{1}{r}{8} &       & 20    & 0.4   &       & 20    & 0.2   &       & 20    & 0.2   &       & 20    & 0.0   &       & 20    & 0.4   &       & 19    & 180.3 &       & 19    & 180.3 &       & 20    & 0.6 \\
			& \multicolumn{1}{r}{10} &       & 20    & 0.3   &       & 20    & 0.1   &       & 20    & 0.1   &       & 20    & 0.0   &       & 20    & 0.3   &       & 20    & 0.2   &       & 20    & 0.2   &       & 20    & 0.5 \\
			& \multicolumn{1}{r}{12} &       & 20    & 0.1   &       & 20    & 0.1   &       & 20    & 0.1   &       & 20    & 0.0   &       & 20    & 0.1   &       & 20    & 0.1   &       & 20    & 0.1   &       & 20    & 0.3 \\
			\midrule
			\multirow{5}[2]{*}{100} & \multicolumn{1}{r}{3} &       & 20    & 62.7  &       & 20    & 59.8  &       & 20    & 64.5  &       & 20    & 0.2   &       & 20    & 23.0  &       & 20    & 37.6  &       & 20    & 34.0  &       & 20    & 9.4 \\
			& \multicolumn{1}{r}{5} &       & 19    & 210.3 &       & 20    & 31.5  &       & 20    & 32.5  &       & 20    & 0.2   &       & 20    & 56.0  &       & 19    & 286.4 &       & 18    & 415.1 &       & 20    & 8.4 \\
			& \multicolumn{1}{r}{8} &       & 19    & 194.4 &       & 20    & 7.0   &       & 20    & 7.1   &       & 20    & 0.1   &       & 20    & 53.2  &       & 20    & 133.6 &       & 20    & 101.0 &       & 20    & 5.0 \\
			& \multicolumn{1}{r}{10} &       & 19    & 189.3 &       & 20    & 3.7   &       & 20    & 3.8   &       & 20    & 0.1   &       & 20    & 41.4  &       & 20    & 18.1  &       & 20    & 17.9  &       & 20    & 4.6 \\
			& \multicolumn{1}{r}{12} &       & 20    & 7.3   &       & 20    & 2.2   &       & 20    & 2.1   &       & 20    & 0.1   &       & 20    & 50.8  &       & 20    & 74.8  &       & 20    & 75.0  &       & 20    & 2.9 \\
			\midrule
			\multirow{5}[2]{*}{150} & \multicolumn{1}{r}{3} &       & 18    & 873.8 &       & 20    & 1162.0 &       & 20    & 1123.4 &       & 20    & 100.4 &       & 20    & 362.3 &       & 15    & 1590.9 &       & 17    & 1147.7 &       & 20    & 127.0 \\
			& \multicolumn{1}{r}{5} &       & 20    & 323.8 &       & 20    & 813.9 &       & 20    & 813.6 &       & 20    & 100.2 &       & 16    & 1438.2 &       & 16    & 1485.2 &       & 15    & 1656.8 &       & 20    & 137.7 \\
			& \multicolumn{1}{r}{8} &       & 18    & 475.6 &       & 20    & 275.9 &       & 20    & 261.5 &       & 20    & 100.2 &       & 17    & 1035.5 &       & 16    & 1641.0 &       & 16    & 1624.3 &       & 20    & 126.0 \\
			& \multicolumn{1}{r}{10} &       & 20    & 85.5  &       & 20    & 117.4 &       & 20    & 132.1 &       & 20    & 100.1 &       & 16    & 1273.3 &       & 16    & 1183.7 &       & 17    & 1168.7 &       & 20    & 110.0 \\
			& \multicolumn{1}{r}{12} &       & 20    & 52.8  &       & 20    & 53.0  &       & 20    & 54.8  &       & 20    & 100.1 &       & 18    & 599.4 &       & 20    & 348.9 &       & 20    & 356.4 &       & 20    & 105.8 \\
			\midrule
			\multicolumn{2}{l}{total/avg} &       & 393   & 124.0 &       & 400   & 126.4 &       & 400   & 124.9 &       & 400   & 25.1  &       & 387   & 246.8 &       & 380   & 358.1 &       & 382   & 338.9 &       & 400   & 32.0 \\
			\bottomrule
		\end{tabular}%
	}
	\label{tab:resnew_Class_I_II}%
\end{table}%

\begin{table}[htbp]
	\centering
	\caption{\rev{Results on set 2 instances -- Classes III and IV}}
	\scriptsize
	\setlength{\tabcolsep}{0.5mm}
	\rev{
		\begin{tabular}{llrrrrrrrrrrrrrrrrrrrrrrrr}
			\toprule
			\multirow{3}[6]{*}{$n$} & \multicolumn{1}{r}{\multirow{3}[6]{*}{$m$}} &       & \multicolumn{11}{c}{Class III} &       & \multicolumn{11}{c}{Class IV} \\
			\cmidrule{4-14}\cmidrule{16-26}
			&       &       & \multicolumn{2}{c}{RF-ZDD} &       & \multicolumn{2}{c}{VHV-ZDD} &       & \multicolumn{2}{c}{VHV-DP} &       & \multicolumn{2}{c}{ILS+EAF} &       & \multicolumn{2}{c}{RF-ZDD} &       & \multicolumn{2}{c}{VHV-ZDD} &       & \multicolumn{2}{c}{VHV-DP} &       & \multicolumn{2}{c}{ILS+EAF} \\
			\cmidrule{4-5}\cmidrule{7-8}\cmidrule{10-11}\cmidrule{13-14}\cmidrule{16-17}\cmidrule{19-20}\cmidrule{22-23}\cmidrule{25-26}
			&       &       & \#opt & t(s)  &       & \#opt & t(s)  &       & \#opt & t(s)  &       & \#opt & t(s)  &       & \#opt & t(s)   &       & \#opt & t(s)   &       & \#opt & t(s)   &       & \#opt & t(s) \\
			\midrule
			\multirow{5}[2]{*}{20} & \multicolumn{1}{r}{3} &       & 20    & 0.0   &       & 20    & 0.0   &       & 20    & 0.0   &       & 20    & 0.0   &       & 20    & 0.0   &       & 20    & 0.0   &       & 20    & 0.0   &       & 20    & 0.0 \\
			& \multicolumn{1}{r}{5} &       & 20    & 0.0   &       & 20    & 0.0   &       & 20    & 0.0   &       & 20    & 0.0   &       & 20    & 0.0   &       & 20    & 0.0   &       & 20    & 0.0   &       & 20    & 0.0 \\
			& \multicolumn{1}{r}{8} &       & 20    & 0.0   &       & 20    & 0.0   &       & 20    & 0.0   &       & 20    & 0.0   &       & 20    & 0.0   &       & 20    & 0.0   &       & 20    & 0.0   &       & 20    & 0.0 \\
			& \multicolumn{1}{r}{10} &       & 20    & 0.0   &       & 20    & 0.0   &       & 20    & 0.0   &       & 20    & 0.0   &       & 20    & 0.0   &       & 20    & 0.0   &       & 20    & 0.0   &       & 20    & 0.0 \\
			& \multicolumn{1}{r}{12} &       & 20    & 0.0   &       & 20    & 0.0   &       & 20    & 0.0   &       & 20    & 0.0   &       & 20    & 0.0   &       & 20    & 0.0   &       & 20    & 0.0   &       & 20    & 0.0 \\
			\midrule
			\multirow{5}[2]{*}{50} & \multicolumn{1}{r}{3} &       & 20    & 2.1   &       & 20    & 4.0   &       & 20    & 4.2   &       & 20    & 0.1   &       & 20    & 0.7   &       & 20    & 0.7   &       & 20    & 1.0   &       & 20    & 0.5 \\
			& \multicolumn{1}{r}{5} &       & 20    & 1.1   &       & 20    & 1.9   &       & 20    & 2.0   &       & 20    & 0.1   &       & 20    & 1.2   &       & 20    & 2.1   &       & 20    & 2.6   &       & 20    & 0.2 \\
			& \multicolumn{1}{r}{8} &       & 20    & 0.6   &       & 20    & 0.6   &       & 20    & 0.6   &       & 20    & 0.1   &       & 20    & 0.4   &       & 20    & 0.5   &       & 20    & 0.6   &       & 20    & 0.1 \\
			& \multicolumn{1}{r}{10} &       & 20    & 0.4   &       & 20    & 0.4   &       & 20    & 0.5   &       & 20    & 0.1   &       & 20    & 0.4   &       & 20    & 0.5   &       & 20    & 0.6   &       & 20    & 0.1 \\
			& \multicolumn{1}{r}{12} &       & 20    & 0.2   &       & 20    & 0.2   &       & 20    & 0.2   &       & 20    & 0.1   &       & 20    & 0.1   &       & 20    & 0.1   &       & 20    & 0.2   &       & 20    & 0.1 \\
			\midrule
			\multirow{5}[2]{*}{100} & \multicolumn{1}{r}{3} &       & 20    & 41.5  &       & 20    & 158.8 &       & 20    & 162.3 &       & 20    & 1.0   &       & 20    & 69.7  &       & 20    & 186.7 &       & 20    & 169.1 &       & 20    & 17.5 \\
			& \multicolumn{1}{r}{5} &       & 20    & 22.4  &       & 20    & 109.5 &       & 20    & 103.7 &       & 20    & 0.8   &       & 20    & 33.3  &       & 20    & 147.8 &       & 20    & 159.3 &       & 20    & 11.6 \\
			& \multicolumn{1}{r}{8} &       & 20    & 10.7  &       & 20    & 40.8  &       & 20    & 39.3  &       & 20    & 0.4   &       & 20    & 14.8  &       & 20    & 44.3  &       & 20    & 47.1  &       & 20    & 1.2 \\
			& \multicolumn{1}{r}{10} &       & 20    & 7.1   &       & 20    & 20.5  &       & 20    & 20.0  &       & 20    & 0.3   &       & 20    & 9.7   &       & 20    & 31.5  &       & 20    & 34.8  &       & 20    & 1.0 \\
			& \multicolumn{1}{r}{12} &       & 20    & 6.1   &       & 20    & 12.2  &       & 20    & 12.4  &       & 20    & 0.2   &       & 19    & 186.9 &       & 20    & 18.4  &       & 20    & 20.4  &       & 20    & 0.4 \\
			\midrule
			\multirow{5}[2]{*}{150} & \multicolumn{1}{r}{3} &       & 20    & 430.2 &       & 20    & 1508.3 &       & 20    & 1517.5 &       & 20    & 100.6 &       & 20    & 973.5 &       & 16    & 2529.2 &       & 20    & 2108.4 &       & 20    & 108.7 \\
			& \multicolumn{1}{r}{5} &       & 20    & 166.6 &       & 20    & 965.5 &       & 20    & 1063.9 &       & 20    & 100.5 &       & 20    & 383.7 &       & 20    & 1281.3 &       & 20    & 1239.4 &       & 20    & 104.0 \\
			& \multicolumn{1}{r}{8} &       & 20    & 71.9  &       & 20    & 504.8 &       & 20    & 517.1 &       & 20    & 100.3 &       & 20    & 164.3 &       & 20    & 612.9 &       & 20    & 648.8 &       & 20    & 101.9 \\
			& \multicolumn{1}{r}{10} &       & 20    & 51.1  &       & 20    & 360.9 &       & 20    & 361.6 &       & 20    & 100.4 &       & 20    & 93.3  &       & 20    & 358.6 &       & 20    & 429.6 &       & 20    & 101.8 \\
			& \multicolumn{1}{r}{12} &       & 20    & 36.8  &       & 20    & 224.4 &       & 20    & 245.2 &       & 20    & 100.2 &       & 20    & 69.0  &       & 19    & 458.2 &       & 20    & 287.1 &       & 20    & 101.0 \\
			\midrule
			\multicolumn{2}{l}{total/avg} &       & 400   & 42.4  &       & 400   & 195.6 &       & 400   & 202.5 &       & 400   & 25.3  &       & 399   & 100.1 &       & 395   & 283.6 &       & 400   & 257.4 &       & 400   & 27.5 \\
			\bottomrule
		\end{tabular}%
	}
	\label{tab:resnew_Class_III_IV}%
\end{table}%

\begin{table}[htbp]
	\centering
	\caption{\rev{Results on set 2 instances -- Classes V and VI}}
	\scriptsize
	\setlength{\tabcolsep}{0.5mm}
	\rev{
		\begin{tabular}{llrrrrrrrrrrrrrrrrrrrrrrrr}
			\toprule
			\multirow{3}[6]{*}{$n$} & \multicolumn{1}{r}{\multirow{3}[6]{*}{$m$}} &       & \multicolumn{11}{c}{Class V} &       & \multicolumn{11}{c}{Class VI} \\
			\cmidrule{4-14}\cmidrule{16-26}
			&       &       & \multicolumn{2}{c}{RF-ZDD} &       & \multicolumn{2}{c}{VHV-ZDD} &       & \multicolumn{2}{c}{VHV-DP} &       & \multicolumn{2}{c}{ILS+EAF} &       & \multicolumn{2}{c}{RF-ZDD} &       & \multicolumn{2}{c}{VHV-ZDD} &       & \multicolumn{2}{c}{VHV-DP} &       & \multicolumn{2}{c}{ILS+EAF} \\
			\cmidrule{4-5}\cmidrule{7-8}\cmidrule{10-11}\cmidrule{13-14}\cmidrule{16-17}\cmidrule{19-20}\cmidrule{22-23}\cmidrule{25-26}
			&       &       & \#opt & t(s)  &       & \#opt & t(s)  &       & \#opt & t(s)  &       & \#opt & t(s)  &       & \#opt & t(s)   &       & \#opt & t(s)   &       & \#opt & t(s)   &       & \#opt & t(s) \\
			\midrule
			\multirow{5}[2]{*}{20} & \multicolumn{1}{r}{3} &       & 20    & 0.0   &       & 20    & 0.0   &       & 20    & 0.0   &       & 20    & 0.0   &       & 20    & 0.0   &       & 20    & 0.0   &       & 20    & 0.0   &       & 20    & 0.0 \\
			& \multicolumn{1}{r}{5} &       & 20    & 0.0   &       & 20    & 0.0   &       & 20    & 0.0   &       & 20    & 0.0   &       & 20    & 0.0   &       & 20    & 0.0   &       & 20    & 0.0   &       & 20    & 0.0 \\
			& \multicolumn{1}{r}{8} &       & 20    & 0.0   &       & 20    & 0.0   &       & 20    & 0.0   &       & 20    & 0.0   &       & 20    & 0.0   &       & 20    & 0.0   &       & 20    & 0.0   &       & 20    & 0.0 \\
			& \multicolumn{1}{r}{10} &       & 20    & 0.0   &       & 20    & 0.0   &       & 20    & 0.0   &       & 20    & 0.0   &       & 20    & 0.0   &       & 20    & 0.0   &       & 20    & 0.0   &       & 20    & 0.0 \\
			& \multicolumn{1}{r}{12} &       & 20    & 0.0   &       & 20    & 0.0   &       & 20    & 0.0   &       & 20    & 0.0   &       & 20    & 0.0   &       & 20    & 0.0   &       & 20    & 0.0   &       & 20    & 0.0 \\
			\midrule
			\multirow{5}[2]{*}{50} & \multicolumn{1}{r}{3} &       & 20    & 1.5   &       & 20    & 1.8   &       & 20    & 1.9   &       & 20    & 1.2   &       & 20    & 2.3   &       & 20    & 3.4   &       & 20    & 2.7   &       & 20    & 2.3 \\
			& \multicolumn{1}{r}{5} &       & 20    & 1.5   &       & 20    & 2.3   &       & 20    & 3.0   &       & 20    & 0.7   &       & 20    & 1.4   &       & 20    & 1.8   &       & 20    & 1.8   &       & 20    & 1.2 \\
			& \multicolumn{1}{r}{8} &       & 20    & 0.3   &       & 20    & 0.6   &       & 20    & 0.7   &       & 20    & 0.2   &       & 20    & 1.0   &       & 20    & 0.7   &       & 20    & 0.9   &       & 20    & 0.8 \\
			& \multicolumn{1}{r}{10} &       & 20    & 0.4   &       & 20    & 0.4   &       & 20    & 0.6   &       & 20    & 0.1   &       & 20    & 1.1   &       & 20    & 0.7   &       & 20    & 0.7   &       & 20    & 0.4 \\
			& \multicolumn{1}{r}{12} &       & 20    & 0.2   &       & 20    & 0.2   &       & 20    & 0.2   &       & 20    & 0.1   &       & 20    & 1.1   &       & 20    & 1.0   &       & 20    & 0.8   &       & 20    & 0.4 \\
			\midrule
			\multirow{5}[2]{*}{100} & \multicolumn{1}{r}{3} &       & 19    & 314.0 &       & 20    & 359.5 &       & 20    & 290.4 &       & 20    & 46.8  &       & 20    & 78.7  &       & 20    & 242.5 &       & 20    & 200.1 &       & 20    & 73.0 \\
			& \multicolumn{1}{r}{5} &       & 20    & 55.0  &       & 20    & 180.4 &       & 20    & 200.7 &       & 20    & 17.2  &       & 20    & 37.8  &       & 20    & 126.7 &       & 20    & 112.8 &       & 20    & 22.8 \\
			& \multicolumn{1}{r}{8} &       & 19    & 197.8 &       & 18    & 497.5 &       & 18    & 408.9 &       & 20    & 2.3   &       & 20    & 25.2  &       & 20    & 58.0  &       & 20    & 44.6  &       & 20    & 24.9 \\
			& \multicolumn{1}{r}{10} &       & 20    & 12.9  &       & 20    & 35.2  &       & 20    & 39.5  &       & 20    & 3.6   &       & 20    & 22.7  &       & 20    & 32.7  &       & 20    & 32.2  &       & 20    & 8.1 \\
			& \multicolumn{1}{r}{12} &       & 18    & 368.1 &       & 18    & 379.3 &       & 19    & 209.4 &       & 20    & 1.0   &       & 20    & 19.0  &       & 20    & 20.8  &       & 20    & 18.8  &       & 20    & 17.8 \\
			\midrule
			\multirow{5}[2]{*}{150} & \multicolumn{1}{r}{3} &       & 20    & 1538.0 &       & 14    & 3086.7 &       & 20    & 2508.5 &       & 20    & 120.9 &       & 20    & 800.0 &       & 18    & 3028.9 &       & 20    & 2662.2 &       & 20    & 233.1 \\
			& \multicolumn{1}{r}{5} &       & 20    & 584.1 &       & 20    & 1488.4 &       & 19    & 1467.4 &       & 20    & 107.2 &       & 20    & 398.4 &       & 20    & 1751.4 &       & 20    & 1434.7 &       & 20    & 178.5 \\
			& \multicolumn{1}{r}{8} &       & 19    & 498.7 &       & 18    & 1053.2 &       & 18    & 1022.9 &       & 20    & 106.4 &       & 20    & 257.7 &       & 20    & 691.6 &       & 20    & 633.1 &       & 20    & 145.3 \\
			& \multicolumn{1}{r}{10} &       & 20    & 164.4 &       & 20    & 422.1 &       & 20    & 541.1 &       & 20    & 103.1 &       & 20    & 210.4 &       & 20    & 439.3 &       & 20    & 406.3 &       & 20    & 151.8 \\
			& \multicolumn{1}{r}{12} &       & 19    & 272.2 &       & 20    & 291.8 &       & 20    & 402.0 &       & 20    & 101.7 &       & 20    & 192.7 &       & 20    & 303.6 &       & 20    & 265.7 &       & 20    & 135.1 \\
			\midrule
			\multicolumn{2}{l}{total/avg} &       & 394   & 200.5 &       & 388   & 390.0 &       & 394   & 354.9 &       & 400   & 30.6  &       & 400   & 102.5 &       & 398   & 335.2 &       & 400   & 290.9 &       & 400   & 49.8 \\
			\bottomrule
		\end{tabular}%
	}
	\label{tab:resnew_Class_V_VI}%
\end{table}%
	
\section{Conclusions}
\label{sec:conclusion}

In this work, we have proposed pseudo-polynomial arc-flow (AF) formulations to solve the problem of scheduling a set of jobs on a set of identical parallel machines by minimizing the total weighted completion time. A first straight AF model already benefits from the fact that schedules follow a {weighted shortest processing time} (WSPT) rule on each machine. A second enhanced AF model (EAF) improves AF by embedding further reduction techniques. EAF needs on average less than $50\%$ of the variables required by AF, and, in some cases, this number drops to less than $20\%$. Computational experiments showed that EAF is very effective and solves instances with up to $1000$ jobs and $30$ machines, performing much better than direct time-indexed formulations and even advanced branch-and-price methods. Still, troubles might arise when the processing times of the jobs assume large values, and more research is envisaged to solve large instances of this type, possibly involving heuristics and column generation methods.

An interesting future research direction also involves the application of the discussed techniques to problems with release dates and/or setup times. In such cases, optimal solutions might nor respect the WSPT sorting, and thus these problems might be very challenging for AF models. In general, in the scheduling area much work has been done on time-indexed formulations, but the application of AF models is new and could lead to large computational benefits in the solution of many problems.

\section*{Acknowledgments}
This research was funded by the CNPq - Conselho Nacional de Desenvolvimento Cient{\'i}fico e Tecnol{\'o}gico, Brazil, grant No. 234814/2014-4, and by Universit{\`a} degli Studi di Modena e Reggio Emilia, project FAR 2015 - Applicazioni della Teoria dei Grafi nelle Scienze, nell'Industria e nella Societ{\`a}. We thank Kerem B\"{u}lb\"{u}l and Halil {\c{S}en} for kindly providing us with additional computation results, {and Daniel Kowalczyk and Roel Leus for sharing the instances of benchmark set 2.} \rev{We also thank three anonymous referees whose comments greatly improved the quality of the paper.}

\bibliographystyle{mmsbib}
\bibliography{ref}

\end{document}